\documentclass[a4paper,UKenglish,cleveref]{lipics-v2019}
\nolinenumbers

\title{QCSP on Reflexive Tournaments}

\titlerunning{QCSP on Reflexive Tournaments}

\author{Beno\^it Larose}{LACIM, Universit\'e du Qu\'ebec a Montr\'eal, Canada}{blarose@lacim.ca}{}{}
\author{Petar Markovi\'c}{Department of Mathematics and Informatics, University of Novi Sad, Serbia}{pera@dmi.uns.ac.rs}{}{}
\author{Barnaby Martin}{Department of Computer Science, Durham University, UK}{barnaby.d.martin@durham.ac.uk}{}{}
\author{Dani\"el Paulusma}{Department of Computer Science, Durham University, UK}{daniel.paulusma@durham.ac.uk}{}{}
\author{Siani Smith}{Department of Computer Science, Durham University, UK}{siani.smith@durham.ac.uk}{}{}
\author{Stanislav \v{Z}ivn\'y}{Department of Computer Science, University of Oxford, UK}{standa.zivny@cs.ox.ac.uk}{}{Stanislav Zivny was supported by a Royal Society University Research Fellowship. This project has received funding from the European Research Council (ERC) under the European Union’s Horizon 2020 research and innovation programme (grant agreement No 714532). The paper reflects only the authors’ views and not the views of the ERC or the European Commission. The European Union is not liable for any use that may be made of the information contained therein.}

\authorrunning{B. Larose, P. Markovi\'c, B. Martin, D. Paulusma, S. Smith and S. \v{Z}ivn\'y}

\Copyright{Beno\^it Larose, Peter Markovi\'c, Barnaby Martin, Dani\"el Paulusma, Siani Smith and Stanislav \v{Z}ivn\'y}



\category{}
\relatedversion{}
\supplement{}
\acknowledgements{}

\EventEditors{John Q. Open and Joan R. Access}
\EventNoEds{2}
\EventLongTitle{42nd Conference on Very Important Topics (CVIT 2016)}
\EventShortTitle{CVIT 2016}
\EventAcronym{CVIT}
\EventYear{2016}
\EventDate{December 24--27, 2016}
\EventLocation{Little Whinging, United Kingdom}
\EventLogo{}
\SeriesVolume{42}
\ArticleNo{23}

\usepackage{amsmath,amssymb,amsthm}
\usepackage{xspace,xcolor}
\usepackage{tikz}
\usepackage{graphicx}
\usetikzlibrary{fit,automata,arrows}

\usepackage[ruled]{algorithm2e}

\usepackage{microtype}
\usepackage[all]{xy}

\usetikzlibrary{shapes.geometric, arrows.meta, calc, positioning}

\def\romannum{\begingroup
\def\theenumi{\textup{(\roman{enumi})}}
\def\p@enumi{}
\def\labelenumi{\theenumi}
\enumerate}

\newcommand{\NP}[0]{\ensuremath{\mathsf{NP}}}
\newcommand{\PP}[0]{\ensuremath{\mathsf{P}}}

\newcommand{\NL}[0]{\ensuremath{\mathsf{NL}}}

\newcommand{\Pspace}[0]{\ensuremath{\mathsf{Pspace}}}

\newcommand{\QCSP}[0]{\ensuremath{\mathrm{QCSP}}}
\newcommand{\CSP}[0]{\ensuremath{\mathrm{CSP}}}

\newcommand{\HC}[0]{\ensuremath{\mathrm{HC}}}
\newcommand{\TT}[0]{\ensuremath{\mathrm{TT}^*}}

\newcommand{\DC}[0]{\ensuremath{\mathrm{DC}}}
\newcommand{\DCm}[0]{\ensuremath{\mathrm{DC}^*_m}}
\newcommand{\Cyl}[0]{\ensuremath{\mathrm{Cyl}^*}}
\newcommand{\Cylm}[0]{\ensuremath{\mathrm{Cyl}^*_m}}

\newcommand{\Cylmplus}[0]{\ensuremath{\mathrm{Cyl}^{*+}_m}}

\newcommand{\Cylplus}[0]{\ensuremath{\mathrm{Cyl}^{*+}}}
\renewcommand{\phi}{\varphi}
\newcommand{\ignore}[1]{}
\newcommand{\A}{\mathrm{A}}

\newcommand{\G}{\ensuremath{\mathrm{G}}}

\newcommand{\F}{\ensuremath{\mathrm{F}}}
\renewcommand{\H}{\ensuremath{\mathrm{H}}}
\newcommand{\B}{\ensuremath{\mathrm{B}}}
\newcommand{\K}{\ensuremath{\mathrm{K}}}

\begin{document}
\maketitle

\begin{abstract}
We give a complexity dichotomy for the Quantified Constraint Satisfaction Problem $\QCSP(\H)$ when $\H$ is a reflexive tournament. It is well-known that reflexive tournaments can be split into a sequence of strongly connected components $\H_1,\ldots,\H_n$ so that there exists an edge from every vertex of $H_i$ to every vertex of $H_j$ if and only if $i<j$. 
We prove that if $\H$ has both its initial and final strongly connected component (possibly equal) of size~$1$, then $\QCSP(\H)$ is in \NL\ and otherwise $\QCSP(\H)$ is \NP-hard.
\end{abstract}

\newpage

\section{Introduction}

The \emph{Quantified Constraint Satisfaction Problem} QCSP$(\B)$, for a fixed \emph{template} (structure) $\B$, is a popular generalisation of the \emph{Constraint Satisfaction Problem} CSP$(\B)$. In the latter, one asks if a primitive positive sentence (the existential quantification of a conjunction of atoms) $\phi$ is true on $\B$, while in the former this sentence may also have universal quantification\footnote{\textcolor{black}{Typically, primitive positive logic also possesses equality, but these can be propagated out by substitution, or removed in the case $x=x$. In the presence of universal quantification, any atom $x=y$ whose innermost variable is universal is false (unless $x$ and $y$ are the same variable). Other instances of equality may be propagated out as before. It follows that the complexity of QCSP$(\B)$ is not affected by the presence or absence of equality, up to logarithmic space reducability.}}. Much of the theoretical research into (finite-domain\footnote{All structures considered in this article are finite.}) CSPs has been in respect of a complexity classification project \cite{FV98,JBK05}, recently completed by \cite{Bu17,Zh17,Zh20}, in which it is shown that all such problems are either in \PP\ or \NP-complete.
Various methods, including combinatorial (graph-theoretic), logical and universal-algebraic were brought to bear on this classification project, with many remarkable consequences.

Complexity classifications for QCSPs appear to be harder than for CSPs. Indeed, a classification for QCSPs will give a fortiori a classification for CSPs (if $\B \uplus \K_1$ is the disjoint union of $\B$ with an isolated element, then QCSP$(\B \uplus\K_1)$ and CSP$(\B)$ are polynomial-time many-one equivalent). Just as CSP$(\B)$ is always in \NP, so QCSP$(\B)$ is always in \Pspace. However, no polychotomy has been conjectured for the complexities of QCSP$(\B)$, though, until recently, only the complexities \PP, \NP-complete and \Pspace-complete were known. Recent work \cite{ZhukM20} has shown that this complexity landscape is considerably richer, and that dichotomies of the form \PP\ versus \NP-hard (using Turing reductions) might be the sensible place to be looking for classifications.

CSP$(\B)$ may equivalently be seen as the \emph{homomorphism} problem which takes as input a structure~$\A$ and asks if there is a homomorphism from $\A$ to $\B$. The \emph{surjective CSP}, SCSP$(\B)$, is a cousin of CSP$(\B)$ in which one requires that this homomorphism from $\A$ to $\B$ be surjective. From the logical perspective this translates to the stipulation that all elements of $\B$ be used as witnesses to the (existential) variables of the primitive positive input $\phi$. The surjective CSP appears in the literature under a variety of names, including \emph{surjective homomorphism} \cite{BKM12}, \emph{surjective colouring} \cite{GPS12,STACStoTOCT} and \emph{vertex compaction} \cite{Vi13,Vi17}. CSP$(\B)$ and SCSP$(\B)$ have various other cousins: see the survey \cite{BKM12} or, in the specific context of reflexive 
tournaments,~\cite{STACStoTOCT}. The only one we will dwell on here is the \emph{retraction} problem CSP$^c(\B)$ which can be defined in various ways but, in keeping with the present narrative, we could define logically as 
allowing atoms of the form $v=b$ in the input sentence $\phi$ where $b$ is some element of $\B$ (the superscript $c$ indicates that constants are allowed). It has only recently been shown that there exists a $\B$ so that SCSP$(\B)$ is in \PP\ while CSP$^c(\B)$ is \NP-complete \cite{zhuk2020norainbow}. It is still not known whether such an example exists among the (partially reflexive) graphs.

It is well-known that the 
binary \emph{cousin} relation 
is not transitive, so let us ask the question as to whether the surjective CSP and QCSP are themselves cousins? The algebraic operations pertaining to the CSP are \emph{polymorphisms} and for QCSP these become \emph{surjective} polymorphisms. On the other hand, a natural use of universal quantification in the QCSP might be to ensure some kind of surjective map (at least under some evaluation of many universally quantified variables). So it is that there may appear to be some relationship between the problems. Yet, there are known irreflexive graphs $\H$ for which QCSP$(\H)$ is in \NL, while SCSP$(\H)$ is \NP-complete (take the $6$-cycle \cite{CiE2006,Vi17}). On the other hand, one can find a $3$-element $\B$ whose relations are preserved by a \emph{semilattice-without-unit} operation such that both CSP$^c(\B)$ and SCSP$(\B)$ are in \PP\ but QCSP$(\B)$ is \Pspace-complete. We are not aware of examples like this among graphs and it is perfectly possible that for (partially reflexive) graphs $\H$, SCSP$(\H)$ being in \PP\ implies that QCSP$(\H)$ is in P.

Tournaments, both irreflexive and reflexive (and sometimes in between), have played a strong role as a testbed for conjectures and a habitat for classifications, for relatives of the CSP 
both complexity-theoretic~\cite{BHM88,TOCL17,STACStoTOCT} and algebraic~\cite{La06,Wires15}. 
Looking at Table~\ref{fig:context} one can see the last unresolved case is precisely QCSP on reflexive tournaments. This is the case we address in this paper. For irreflexive tournaments $\H$, QCSP$(\H)$ is in \PP\ if and only if SCSP$(\H)$ in P, but for reflexive tournaments this is not the case. When $\H$ is a reflexive tournament, we prove that QCSP$(\H)$ is in 
\NL\
if $\H$ has both initial and final
strongly connected
components trivial, and is \NP-hard otherwise. In contrast to the proof from~\cite{TOCL17} and like the proof of \cite{STACStoTOCT}, we will henceforth work largely combinatorially rather than algebraically. Note that we do not investigate beyond \NP-hard, 
so our dichotomy cannot be compared directly to the trichotomy of~\cite{TOCL17} 
for irreflexive tournaments
which distinguishes between \PP, \NP-complete and \Pspace-complete.

\begin{table}[h]
\begin{center}
\resizebox{!}{1.05cm}{
\begin{tabular}{|p{1.9cm}|p{2.4cm}|p{3cm}|p{3cm}|p{3cm}|}
\hline
& QCSP & CSP & Surjective CSP & Retraction \\
\hline
irreflexive tournaments & trichotomy \cite{TOCL17} & dichotomy \cite{BHM88} & dichotomy \cite{BHM88} & dichotomy \cite{BHM88} \\
\hline
reflexive\hspace{1cm} tournaments & {\bf this paper} & all trivial  & dichotomy \cite{STACStoTOCT} & dichotomy \cite{La06} \\
\hline
\end{tabular}
}
\end{center}
\caption{Our result in a wider context. The results for irreflexive tournaments were all proved in the more general setting of irreflexive semicomplete digraphs in the papers cited.}
\label{fig:context}
\end{table}
In Section~\ref{s-tour} we prove the \NP-hard cases of our dichotomy.
Our proof method follows that from \cite{STACStoTOCT}, while adapting the ideas of \cite{LMCS2015} in order to make what was developed for Surjective CSP applicable to QCSP. 
The QCSP
 is not naturally a combinatorial problem but can be seen thusly when viewed in 
 a
 certain way. 
 We indeed closely mirror \cite{STACStoTOCT} with \cite{LMCS2015} in the strongly connected case. For the not strongly connected case, the adaptation from the strongly connected case was straightforward for the Surjective CSP in \cite{STACStoTOCT}. However, the straightforward method does not work for the QCSP. Instead, we seek a direct argument that essentially sees us extending the method from \cite{STACStoTOCT}. 
 
In Section~\ref{s-nl} we prove the \NL\ cases of our dichotomy. Here, we use ideas originally developed in (the conference version of) \cite{LMCS2015} and first used in the wild in \cite{Martin11}. Thus, we do not introduce new proof techniques as such but rather weave our proof through the reasonably intricate synthesis of different known techniques. 
In Section~\ref{s-con} we state our dichotomy and give some directions for future work. 

\section{Preliminaries}

For an integer $k\geq 1$, we write $[k]:=\{1,\ldots,k\}$. 
A  vertex $u\in V(G)$ in a digraph $G$ is {\it backwards-adjacent} to another vertex $v\in V$ if $(u,v)\in E$. It is {\it forwards-adjacent} to another vertex $v\in V$ if $(v,u)\in E$.
If a vertex $u$ has a self-loop $(u,u)$, then $u$ is \emph{reflexive}; otherwise $u$ is \emph{irreflexive}. A digraph $G$ is \emph{reflexive} or \emph{irreflexive} if all its vertices are reflexive or irreflexive, respectively.

The \emph{directed path} on $k$ vertices is the digraph with vertices $u_0,\ldots, u_{k-1}$ and edges $(u_i,u_{i+1})$ for $i=0,\ldots,k-2$. By adding the edge $(u_{k-1},u_0)$, we obtain
the \emph{directed cycle} on $k$ vertices.
A digraph $\G$ is \emph{strongly connected} if for all $u,v \in V(\G)$ there is a 
directed path 
in $E(\G)$ from $u$ to $v$.
A \emph{double edge} in a digraph~$\G$ consists in a pair of distinct vertices $u,v \in V(\G)$, so that $(u,v)$ and $(v,u)$ belong to $E(\G)$. 
A digraph~$\G$ is \emph{semicomplete} if for every two distinct vertices $u$ and $v$, at least one of $(u,v)$, $(v,u)$ belongs to $E(\G)$. A semicomplete digraph~$\G$ is a \emph{tournament} if for every two distinct vertices $u$ and $v$, exactly one of $(u,v)$, $(v,u)$ belongs to $E(\G)$. 
A reflexive tournament~$\G$ is \emph{transitive} if for every three vertices $u,v,w$ with $(u,v), (v,w)\in E(\G)$, also $(u,w)$ belongs to $E(\G)$.
 A digraph~$\F$ is a \emph{subgraph} of a digraph~$\G$ if $V(\F) \subseteq V(\G)$ and $E(\F) \subseteq E(\G)$. It is \emph{induced} if $E(\F)$ coincides with $E(\G)$ restricted to pairs containing only vertices of $V(\F)$. A \emph{subtournament} is an induced subgraph of a tournament.
 It is well known that a reflexive tournament $\H$ can be split into a sequence of strongly connected components $\H_1,\ldots,\H_n$ for some integer $n\geq 1$ so that there exists an edge from every vertex of $H_i$ to every vertex of $H_j$ if and only if $i<j$. We will use the notation $\H_1\Rightarrow \cdots \Rightarrow \H_n$ for $\H$ and we refer to $\H_1$ and $\H_n$ as the {\it initial} and {\it final} components of $\H$, respectively.

A \emph{homomorphism} from a digraph~$\G$ to a digraph~$\H$ is a function $f:V(\G)\rightarrow V(\H)$ such that for all $u,v \in V(\G)$ with $(u,v) \in E(\G)$ we have $(f(u),f(v)) \in E(\H)$. We say that $f$ is \emph{(vertex)-surjective} if for every vertex $x\in V(\H)$ there exists a vertex $u\in V(\G)$ with $f(u)=x$.
A digraph $\H'$ is a \emph{homomorphic image} of a digraph~$\H$ if 
there is a surjective homomorphism from $\H$ to $\H'$ that is also {\it edge-surjective}, that is, for all $(x',y') \in E(\H')$ there exists an $(x,y) \in E(\H)$ with $x'=h(x)$ and $y'=h(y)$. 

The problem {\sc $\H$-Retraction} takes as input a graph $\G$ of which $\H$ is an induced subgraph and asks whether there is a homomorphism from $\G$ to $\H$ that is the identity on $\H$. This definition is polynomial-time many-one equivalent to the one we suggested in the introduction (see \mbox{e.g.} \cite{BKM12}). The \emph{quantified constraint satisfaction problem} $\QCSP(\H)$ takes as input a sentence $\phi:=\forall x_1 \exists y_1 \ldots \forall x_n \exists y_n \ \Phi(x_1,y_1,\ldots,x_n,y_n)$, where $\Phi$ is a conjunction of positive atomic (binary edge) relations. This is a yes-instance to the problem just in case $\H \models \phi$.

The \emph{canonical query} of $\G$ (from \cite{KolaitisVardiBook05}) is a primitive positive sentence $\phi_{\G}$ that has the property that, for all $\H$, $\G$ has a homomorphism to $\H$ iff $\H \models \phi_{\G}$. It is built by mapping edges $(x,y)$ from $E(\G)$ to atoms $E(x,y)$ is an existentially quantified conjunctive query.

The \emph{direct product} of two digraphs $\G$ and $\H$, denoted $\G \times \H$, is the digraph on vertex set $V(\G) \times V(\H)$ with edges $((x,y),(x',y'))$ if and only if $(x,x') \in E(\G)$ and $(y,y') \in E(\H)$. We denote the direct product of $k$ copies of $G$ by $G^k$.
A $k$-ary \emph{polymorphism} of~$\G$ is a homomorphism $f$ from $G^k$ to $G$; if $k=1$, then $f$ is also called an
\emph{endomorphism}.
A $k$-ary polymorphism $f$ is \emph{essentially unary} if there exists a unary operation $g$ and $i \in [k]$ so that $f(x_1,\ldots,x_k)=g(x_i)$
for every $(x_1,\ldots,x_k)\in \G^k$.
Let us say that a $k$-ary polymorphism~$f$ is {\it uniformly} $z$ for some $z\in V(\G)$ if $f(x_1,\ldots,x_k)=z$ for every $(x_1,\ldots,x_k) \in V(\G^k)$.
We need the following two lemmas.

\begin{lemma}
Let $H$ be a reflexive tournament and $f$ be a $k$-ary polymorphism of $\H$.  If $f(x,\ldots,x)=z$ for every $x \in V(\H)$, then $f$ is uniformly \textcolor{black}{equal to} $z$.
\label{lem:uniformly-z}
\end{lemma}
\begin{proof}
Consider some tuple $(x_1,\ldots,x_k)$ which has $m$ distinct vertices. We proceed by induction on $m$, where the base case $m=1$ is given as an assumption. Suppose we have the result for $m$ vertices and let $(x_1,\ldots,x_k)$ have $m+1$ distinct entries. For simplicity (and w.l.o.g.) we will consider this reordered and without duplicates as $(y_1,\ldots,y_m,y_{m+1})$.  Suppose $f$ maps $(x_1,\ldots,x_k)$ to $z'$. Assume $(y_{m},y_{m+1}) \in E(\H)$ (the case $(y_{m+1},y_{m})$ is symmetric). Then consider the tuples $(y_1,\ldots,y_{m},y_{m})$ and $(y_1,\ldots,y_{m+1},y_{m+1})$. By the inductive hypothesis, $f$ maps each of these (when reordered and padded appropriately with duplicates) to $z$. Furthermore, we have co-ordinatewise edges from $(y_1,\ldots,y_{m},y_{m})$ to $(y_1,\ldots,y_{m},y_{m+1})$ and from $(y_1,\ldots,y_{m},y_{m+1})$ to $(y_1,\ldots,y_{m+1},y_{m+1})$. Since we deduce by the definition of polymorphism that both $(z,z'), (z',z) \in E(\H)$, it follows that $z'=z$. Thus, $f$ maps also $(y_1,\ldots,y_m,y_{m+1})$  (when reordered and padded appropriately with duplicates) to $z$. That is, $f(x_1,\ldots,x_k)=z$.
\end{proof}

\begin{lemma}
Let $\H$ be the reflexive tournament $\H_1\Rightarrow \cdots \Rightarrow \H_i \Rightarrow \cdots \Rightarrow \H_n$. If $f$ is a $k$-ary surjective polymorphism of $\H$, then $f$ preserves each of $V(\H_1),\ldots,V(\H_n)$;
that is, for every $i$ and every tuple of $k$ vertices $x_1,\ldots,x_k\in V(\H_i)$, $f(x_1,\ldots,x_k)\in V(\H_i)$.
\end{lemma}
\begin{proof}
Suppose $f$ maps some tuple $(x_1,\ldots,x_m)$ from $V(\H_i)$ to $y \in V(\H_\ell)$. Let $(x'_1,\ldots,x'_m)$ be any tuple from $V(\H_i)$. Since $\H_i$ is strongly connected, $f(x'_1,\ldots,x'_m)$ in $V(\H_\ell)$. It follows that if $\ell\neq i$, \mbox{e.g. w.l.o.g.} $\ell<i$, then some component $\ell'\geq i$ can not be in the range of $f$.
\end{proof}
The relevance of this lemma is in its sequent corollary, which follows according to Proposition 3.15 of \cite{BBCJK}.
\begin{corollary}
Let $\H$ be the reflexive tournament $\H_1\Rightarrow \cdots \Rightarrow \H_i \Rightarrow \cdots \Rightarrow \H_n$. Each subset of the domain $V(\H_i)$ is definable by a QCSP instance in one free variable.
\label{cor:galois}
\end{corollary}
An endomorphism~$e$ of a digraph~$\G$ is a \emph{constant map} if there exists a vertex $v\in V(\G)$ such that $e(u)=v$ for every $u\in V(\G)$, and $e$ is the {\it identity} if $e(u)=u$ for every $u\in \G$.
An \emph{automorphism} is a bijective endomorphism whose inverse is a homomorphism.
An endomorphism is \emph{trivial} if it is either an automorphism or a constant map; otherwise it is {\it non-trivial}.
A digraph is \emph{endo-trivial} if all of its endomorphisms are trivial. 
An endomorphism~$e$ of a digraph~$\G$ \emph{fixes} a subset $S \subseteq V(\G)$ if $e(S)=S$, that is, $e(x)\in S$ for 
every $x \in S$, 
and $e$ fixes an induced subgraph~$\F$ of~$\G$ if it is the identity on $V(\F)$.
It fixes an induced subgraph~$\F$ \emph{up to automorphism} if $e(\F)$ is an automorphic copy of~$\F$. 
An endomorphism~$e$ of~$\G$ is a \emph{retraction} of~$\G$ if $e$ is the identity on~$e(V(\G))$. A digraph is \emph{retract-trivial} if all of its retractions are the identity or constant maps. Note that endo-triviality implies retract-triviality, but the reverse implication is not necessarily true (see \cite{STACStoTOCT}). However, on reflexive tournaments both concepts do coincide~\cite{STACStoTOCT}.

We need a series of results from~\cite{STACStoTOCT}. The third one follows from the well-known fact that every strongly connected tournament has a directed Hamilton cycle~\cite{Ca59}. 

\begin{lemma}[\cite{STACStoTOCT}]
A reflexive tournament is endo-trivial if and only if it is retract-trivial.
\label{lem:endo-retraction-trivial}
\end{lemma}

\begin{lemma}[\cite{STACStoTOCT}]\label{l-Benoit}
Let $\H$ be an endo-trivial reflexive digraph with at least three vertices. Then every polymorphism of  $\H$ is essentially unary.  
\end{lemma}

\begin{lemma}[\cite{STACStoTOCT}]\label{l-Ham}
If $\H$ is an endo-trivial reflexive tournament, then $\H$ contains a directed Hamilton cycle.
\end{lemma}

\begin{lemma}[\cite{STACStoTOCT}]
If $\H$ is an endo-trivial reflexive tournament, then every homomorphic image of~$\H$ of size $1<n<|V(\H)|$ has a double edge.
\label{lem:under-m}
\end{lemma}
\begin{corollary}\label{c-main}
If $\H$ is an endo-trivial reflexive digraph on at least three vertices, then $\QCSP(\H)$ is \NP-hard (in fact it is even \Pspace-complete). 
\end{corollary}
\begin{proof}
This follows from Lemma~\ref{l-Benoit} and \cite{BBCJK}.
\end{proof}

\section{The Proof of the NP-Hard Cases of the Dichotomy}\label{s-tour}


We commence with the \NP-hard cases of the dichotomy. The simpler \NL\ cases will follow, \textcolor{black}{in Section~\ref{s-nl}}. \textcolor{black}{In this section, the central results will appear as Corollaries~\ref{cor:strongly-connected1} and \ref{cor:initial-full}. However, each of these proceeds via an induction where there are two base cases and two inductive (general) cases. Thus, there are eight principal propositions. Propositions~\ref{prop:main1}, \ref{prop:main2}, \ref{prop:main-general1} and \ref{prop:main-general2} lead to Corollary~\ref{cor:strongly-connected1} and Propositions~\ref{prop:mainA1}, \ref{prop:mainA2}, \ref{prop:main-generala-1} and \ref{prop:main-general-a2} lead to Corollary~\ref{cor:initial-full}. The base cases are the simplest to understand and are given in the most detail. The principal propositions commence in Section~\ref{sec:principal-lemma}. Before this we introduce our construction with some supporting lemmas.
}

\subsection{The NP-Hardness Gadget}
We introduce the gadget $\Cylm$ from \cite{STACStoTOCT} drawn in Figure~\ref{fig:Photo2}. Take $m$ disjoint copies of the (reflexive) directed $m$-cycle $\DCm$ arranged in a cylindrical fashion so that there is an edge from $i$ in the $j$th copy to $i$ in the $(j+1)$th copy (drawn in red), and an edge from $i$ in the $(j+1)$th copy to $(i+1) \bmod m$ in the $j$th copy (drawn in green). We consider $\DCm$ to have vertices $\{1,\ldots,m\}$.
Recall that every strongly connected (reflexive) tournament on $m$ vertices has a Hamilton Cycle $\HC_m$.
We label the vertices of $\HC_m$ as $1,\ldots,m$ in order to attach it to the gadget $\Cylm$.\footnote{The superscripted $*$ indicates that the corresponding graph is reflexive. This notation is inherited from \cite{STACStoTOCT}. It is not significant since we could safely assume every graph we work with is reflexive as the template is a reflexive tournament.}
%

\begin{figure}
\centering

\input{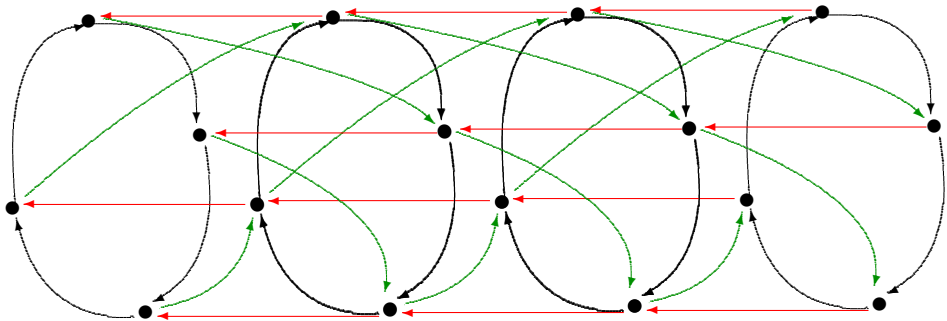}
\caption{The gadget $\Cylm$ in the case $m:=4$ (self-loops are not drawn). We usually visualise the right-hand copy of~$\DC^*_4$ as the ``bottom'' copy and then we talk about vertices ``above'' and ``below'' according to the red arrows.}
\label{fig:Photo2}
\end{figure}

The following lemma follows from induction on the copies of~$\DCm$, since a reflexive tournament has no double edges.
\begin{lemma}[\cite{STACStoTOCT}]
In any homomorphism $h$ from $\Cylm$, with bottom cycle~$\DCm$, to a reflexive tournament, if $|h(\DCm)|=1$, then $|h(\Cylm)|=1$.
\label{lem:claim1}
\end{lemma}

We will use another property, denoted \textbf{$(\dagger)$}, of~$\Cylm$, which is that the retractions from~$\Cylm$ to its bottom copy of~$\DCm$, once propagated through the intermediate copies, induce on the top copy precisely the set of automorphisms of~$\DCm$. That is, the top copy of~$\DCm$ is mapped isomorphically to the bottom copy, and all such isomorphisms may be realised. The reason is that in such a retraction, the $(j+1)$th copy may either map under the identity to the $j$th copy, or rotate one edge of the cycle clockwise,  and $\Cylm$ consists of sufficiently many (namely $m$) copies of~$\DCm$.
Now let $\H$ be a reflexive tournament that contains a subtournament~$\H_0$ on $m$ vertices that is endo-trivial.
By Lemma~\ref{l-Ham}, we find that $\H_0$ contains at least one directed Hamilton cycle $\HC_0$. 
Define $\mathrm{Spill}_m(\H[\H_0,\HC_0])$ as follows. Begin with $\H$ and add a copy of the gadget $\Cylm$, where the bottom copy of~$\DCm$ is identified with $\HC_0$, to build a 
digraph~$\F(\H_0,\HC_0)$. Now ask, for some $y \in V(\H)$ whether there is a retraction~$r$ of~$\F(\H_0,\HC_0)$ to~$\H$ so that some vertex~$x$ (not dependent on $y$) in the top copy of~$\DCm$ in~$\Cylm$ is such that $r(x)=y$. Such vertices~$y$ comprise the set $\mathrm{Spill}_m(\H[\H_0,\HC_0])$.

\medskip
\noindent
{\bf Remark~1.}
If $x$ belongs to some copy of~$\DCm$ that is not the top copy, we can find a vertex~$x'$ in the top copy of~$\DCm$ and a retraction $r'$ from $\F(\H_0,\HC_0)$ to $\H$ with $r'(x')=r(x)=y$, namely by letting $r'$ map the vertices of higher copies of~$\DCm$ to the image of their corresponding vertex in the copy that contains~$x$. In particular this implies that $\mathrm{Spill}_m(\H[\H_0,\HC_0])$ contains $V(\H_0)$.

\medskip
\noindent
We note that the set $\mathrm{Spill}_m(\H[\H_0,\HC_0])$ is potentially dependent on which Hamilton cycle in $\H_0$ is chosen. We now recall that $\mathrm{Spill}_m(\H[\H_0,\HC_0])=V(\H)$ if $\H$ retracts to~$\H_0$.


\begin{lemma}[\cite{STACStoTOCT}]
If $\H$ is a reflexive tournament that retracts to a 
subtournament~$\H_0$ with Hamilton cycle~$\HC_0$, then $\mathrm{Spill}_m(\H[\H_0,\HC_0])=V(\H)$.
\label{lem:spillage}
\end{lemma}


We now review a variant of a construction from \cite{LMCS2015}. Let $\G$ be a graph containing $\H$ where $|V(\H)|$ is of size $n$. Consider all possible functions $\lambda:[n]\rightarrow V(\H)$ (let us write $\lambda \in V(\H)^{[n]}$ of cardinality $N$). For some such $\lambda$, let $\mathcal{G}(\lambda)$ be the graph $G$ enriched with constants $c_1,\ldots,c_n$ where these are interpreted over $V(\H)$ according to $\lambda$ in the natural way (acting on the subscripts). We use calligraphic notation to remind the reader the signature has changed from $\{E\}$ to $\{E,c_1,\ldots,c_n\}$ but we will still treat these structures as graphs. If we write $\G(\lambda)$ without calligraphic notation we mean we look at only the $\{E\}$-reduct, that is, we drop the constants. Of course, $\G(\lambda)$ will always be $\G$.

Let $\mathcal{G}=\bigotimes_{\lambda \in V(\H)^{[n]}} \mathcal{G}(\lambda)$. That is, the vertices of $\mathcal{G}$ are $N$-tuples over $V(\G)$ and there is an edge between two such vertices $(x_1,\ldots,x_N)$ and $(y_1,\ldots,y_N)$ if and only if $(x_1,y_1),\ldots,(x_N,y_N) \in E(\G)$. Finally, the constants $c_i$ are interpreted as $(x_1,\ldots,x_N)$ so that $\lambda_1(c_i)=x_1, \ldots, \lambda_N(c_i)=x_N$. An important induced substructure of $\mathcal{G}$ is $\{(x,\ldots,x):x \in V(\G)\}$. It is a copy of $\G$ called the \emph{diagonal} copy and will play an important role in the sequel. To comprehend better the construction of $\mathcal{G}$ from the sundry $\mathcal{G}(\lambda)$, confer on Figure~\ref{fig:dir-prod}.

The final ingredient of our fundamental construction involves taking some structure $\mathcal{G}$ and making its canonical query with all vertices other than those corresponding to $c_1,\ldots,c_n$ becoming existentially quantified variables (as usual in this construction). We then turn the $c_1,\ldots,c_n$ to variables $y_1,\ldots,y_n$ to make $\phi_\mathcal{G}(y_1,\ldots,y_n)$. Let $\mathcal{H}$ come from the given construction in which $G=H$. It is proved in \cite{LMCS2015} that $\H' \models \forall y_1,\ldots,y_n \ \phi_{\mathcal{H}}(y_1,\ldots,y_n)$ if and only if $\QCSP(\H) \subseteq \QCSP(\H')$ (here we identify $\QCSP(\H)$ with the set of sentences that form its yes-instances). By way of a side note, let us consider a $k$-ary relation $R$ over $\H$ with tuples $(x^1_1,\ldots,x^1_k)$, \ldots, $(x^r_1,\ldots,x^r_k)$. For $i \in [r]$, let $\lambda_i$ map $(c_1,\ldots,c_k)$ to $(x^i_1,\ldots,x^i_k)$. Let $\mathcal{H}=\bigotimes_{\lambda \in \{\lambda_1,\ldots,\lambda_r\}} \mathcal{H}(\lambda)$. Then $\phi_{\mathcal{H}}(y_1,\ldots,y_n)$ is the closure of $R$ under the polymorphisms of $\H$.

\begin{figure}
\begin{center}
$
\xymatrix{
& c_2 \ar[dr] & \\
c_1 \ar[ur] & & \bullet \ar[ll] \\
}
$
\
\resizebox{!}{2cm}{$\begin{array}{c} \mbox{\ } \\ \times \end{array}$}
\
$
\xymatrix{
& \bullet \ar[dr] & \\
c_2  \ar[ur] & & c_1 \ar[ll] \\
}
$
\
\resizebox{!}{2cm}{$\begin{array}{c} \mbox{\ } \\ = \end{array}$}
\
$
\xymatrix{
\bullet \ar[dr] & c_2 \ar[dr] & \bullet \ar[dll] \\
\bullet \ar[dr] & \bullet \ar[dr] & \bullet \ar[dll] \\
c_1 \ar[uur]  & \bullet \ar[uur] & \bullet \ar@/^/[uull] \\
}
$

$
\xymatrix{
& c_2 \ar[dr] & \\
c_1 \ar[ur] & & \bullet \ar[ll] \\
}
$
\
\resizebox{!}{2cm}{$\begin{array}{c} \mbox{\ } \\ \times \end{array}$}
\
$
\xymatrix{
& c_2 \ar[dr] & \\
\bullet \ar[ur] & & c_1 \ar[ll] \\
}
$
\
\resizebox{!}{2cm}{$\begin{array}{c} \mbox{\ } \\ = \end{array}$}
\
$
\xymatrix{
\bullet \ar[dr] & \bullet \ar[dr] & \bullet \ar[dll] \\
\bullet \ar[dr] & c_2 \ar[dr] & \bullet \ar[dll] \\
c_1 \ar[uur]  & \bullet \ar[uur] & \bullet \ar@/^/[uull] \\
}
$
\end{center}
\caption{Illustrations of direct product with constants.}
\label{fig:dir-prod}
\end{figure}
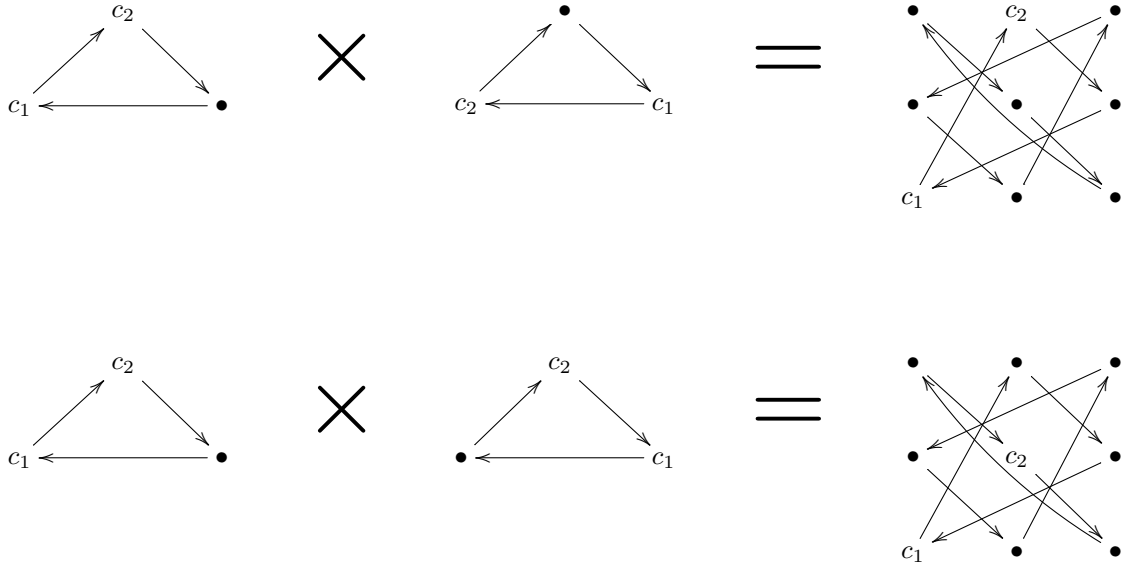

\subsection{The strongly connected case: Two Base Cases}
\label{sec:principal-lemma}

Recall that if $\H$ is a (reflexive) endo-trivial tournament, 
then $\QCSP(\H)$ is \NP-hard due to Lemma~\ref{l-Benoit} combined with 
the results from \cite{BBCJK}. \textcolor{black}{Indeed, Theorem~5.2 in \cite{BBCJK} states that any $\H$ with more than one element, such that all surjective polymorphisms of $\H$ are essentially unary, satisfies that $\QCSP(\H)$ is \Pspace-complete.} However $\H$ may not be endo-trivial. We will now show how to deal with the case
where $\H$ is not endo-trivial but retracts to an endo-trivial subtournament. For doing this we use the \NP-hardness gadget, but we need to distinguish between two different cases. 


\begin{proposition}[Base Case I.]
Let $\H$ be a reflexive tournament that retracts to an endo-trivial subtournament $\H_0$ with Hamilton cycle $\HC_0$. 
Assume that $\H$ retracts to $\H'_0$ for every isomorphic copy $\H'_0=i(\H_0)$ of~$\H_0$ in~$\H$ with $\mathrm{Spill}_m(\H[\H'_0,i(\HC_0)])=V(\H)$.
Then $\H_0$-{\sc Retraction} can be polynomially reduced to $\QCSP(\H)$.
\label{prop:main1}
\end{proposition}

\begin{proof}
Let $m$ be the size of~$|V(\H_0)|$ and $n$ be the size of $|V(\H)|$. 
Let~$\G$ be an instance of $\H_0$-{\sc Retraction}.
We build an instance~$\phi$ of $\QCSP(\H)$
in the following fashion. First, take a copy of~$\H$ together with~$\G$ and build~$\G'$ by identifying these on the copy of~$\H_0$ that they both possess as an 
induced
subgraph.
Now, consider all possible functions $\lambda:[n]\rightarrow V(\H)$. For some such $\lambda$, let $\mathcal{G'}(\lambda)$ be the graph enriched with constants $c_1,\ldots,c_n$ where these are interpreted over some subset of $V(\H)$ according to $\lambda$ in the natural way (acting on the subscripts).

Let $\mathcal{G}'=\bigotimes_{\lambda \in {V(\H)}^{[n]}} \mathcal{G}'(\lambda)$. Let $\G'^d$, $\H^d$ and $\H^d_0$ be the diagonal copies of $\G'$, $\H$ and $\H_0$ in $\mathcal{G}'$. Let $\mathcal{H}$ be the subgraph of $\mathcal{G}'$ induced by $V(\H) \times \cdots \times V(\H)$. Note that the constants $c_1,\ldots,c_n$ live in $\mathcal{H}$. Now build $\mathcal{G}''$ from $\mathcal{G}'$ by augmenting a new copy 
of~$\Cylm$ for every vertex $v \in V(\mathcal{H}) \setminus V(\H^d_0)$. Vertex~$v$ is to be identified with any vertex in the top copy of~$\DCm$ in~$\Cylm$ and the bottom copy of~$\DCm$ is to be identified with $\HC_0$ in $\H^d_0$ according to the identity function. (Thus, in each case, the new vertices are the middle cycles of $\Cylm$ and all but one of the vertices in the top cycle of $\Cylm$.)

Finally, build $\phi$ from the canonical query of $\mathcal{G}''$ where we additionally turn the constants $c_1,\ldots,c_n$ to outermost universal variables. The size of $\phi$ is doubly exponential in $n$ (the size of $H$) but this is constant, so still polynomial in the size of $G$.


We claim that~$\G$ retracts to~$\H_0$ if and only if $\phi \in \QCSP(\H)$.

First suppose that~$\G$ retracts to~$\H_0$. Let $\lambda$ be some assignment of the universal variables of $\phi$ to $\H$. To prove $\phi \in \QCSP(\H)$ it suffices to prove that there is a homomorphism from $\mathcal{G}''$ to $\H$ that extends $\lambda$. Then for this it suffices to prove that there is a homomorphism $h$ from $\mathcal{G}'$ that extends $\lambda$. Let us explain why. Because $\H$ retracts to $\H_0$, we have $\mathrm{Spill}_m(\H[\H_0,\HC_0])=V(\H)$ due to Lemma~\ref{lem:spillage}.  Hence, if $h(x)=y$ for two vertices $x\in V(\mathcal{H})\setminus V(\H^d_0)$ and $y\in V(\H)$, we can always find a retraction of the graph $\F(\H_0,\HC_0)$ to $\H$ that maps~$x$ to~$y$, and we mimic this retraction on the corresponding subgraph in $\mathcal{G}''$. The crucial observation is that this can be done independently for each vertex in $V(\mathcal{H})\setminus V(\H^d_0)$, as two vertices of different copies of $\Cylm$ are only adjacent if they both belong to~$\mathcal{H}$.

Henceforth let us consider the homomorphic image of $\mathcal{G}'$ that is $\mathcal{G}'(\lambda)$. To prove $\phi \in \QCSP(\H)$ it suffices to prove that there is a homomorphism from $\G'(\lambda)$ to $\H$ that extends $\lambda$. Note that it will be sufficient to prove that $\G'$ retracts to $\H$. Let $h$ be the natural retraction from~$\G'$ to~$\H$ that extends the known retraction from $\G$ to $\H_0$. We are done.  

Suppose now $\phi \in \QCSP(\H)$. Choose some surjection for $\lambda$, the assignment of the universal variables of $\phi$ to $\H$. Recall $N=|V(\H)^{[n]}|$. The evaluation of the existential variables that witness $\phi \in \QCSP(\H)$ induces a surjective homomorphism $s$ from $\mathcal{G}''$ to $\H$ which contains within it a surjective homomorphism $s'$ from $\mathcal{H}=\H^N$ to $\H$. Consider the diagonal copy of $\H^d_0 \subset \H^d \subset \G'^d$ in $\mathcal{G}'$. By abuse of notation we will also consider each of $s$ and $s'$ acting just on the diagonal. If $|s'(\H^d_0)|=1$, by construction of $\mathcal{G}''$, we have $|s'(\H^d)|=1$. Indeed, this was the property we noted in Lemma~\ref{lem:claim1}. By Lemma  \ref{lem:uniformly-z}, this would mean $s'$ is uniformly mapping $\mathcal{H}$ to one vertex, which is impossible as $s'$ is surjective. Now we will work exclusively in the diagonal copy $\G'^d$. As $1<|s'(\H^d_0)|<m$ is not possible either due to Lemma~\ref{lem:under-m}, we find that $|s'(\H^d_0)|=m$, and indeed $s'$ maps $\H^d_0$ to a copy of itself in $\H$ which we will call $\H'_0=i(\H^d_0)$ for some isomorphism~$i$. 

We claim that $\mathrm{Spill}_m(\H[\H'_0,i(\HC^d_0)]) = V(\H)$. In order to see this, consider a vertex $y\in V(\H)$. 
As $s'$ is surjective, there exists a vertex~$x\in V(\mathcal{H})$ with $s'(x)=y$. By construction, 
$x$ belongs to some top copy of $\DCm$ in $\Cylm$ in $\F(\H_0,\HC_0)$.
We can extend $i^{-1}$ to an isomorphism from the copy of $\Cylm$  (which has $i(\HC^d_0)$ as its bottom cycle)  in the graph $\F(\H_0',i(\HC^d_0))$ to the copy of $\Cylm$ (which has $\HC^d_0$ as its bottom cycle) in the graph $\F(\H_0,\HC_0)$.
We define a mapping~$r^*$ from $\F(\H_0',i(\HC^d_0))$ to $\H$ by $r^*(u)=s'\circ i^{-1}(u)$ if $u$ is on the copy of $\Cylm$ in $\F(\H_0',i(\HC^d_0))$ and $r^*(u)=u$ otherwise. We observe that $r^*(u)=u$ if $u\in V(\H_0')$ as $s'$ coincides with $i$ on $\H_0$. As $\H^d_0$ separates the other vertices of 
the copy of $\Cylm$ from $V(\H^d)\setminus V(\H^d_0)$, in the sense that removing $\H^d_0$ would disconnect them, this means that $r^*$ is a retraction from $\F(\H_0',i(\HC^d_0))$ to $\H$. We find that $r^*$ maps $i(x)$ to $s'\circ i^{-1}(i(x))=s'(x)=y$. Moreover, as $x$ is in the top copy of $\DCm$ in $\F(\H_0,\HC_0)$,
we conclude that $y$ always belongs to  $\mathrm{Spill}_m(\H[\H'_0,i(\HC^d_0)])$.

As $\mathrm{Spill}_m(\H[\H'_0,i(\HC^d_0)]) = V(\H)$, we find, by assumption of the lemma, that there exists a retraction $r$ from $\H$ to $\H'_0$.  
Now, recalling that we can view $s'$ acting just on the diagonal copy $\H^d$ of $\H$, $i^{-1} \circ r \circ s'$ is the desired retraction of $\G$ to~$\H_0$.  
\end{proof}

\begin{figure}
\centering
\input{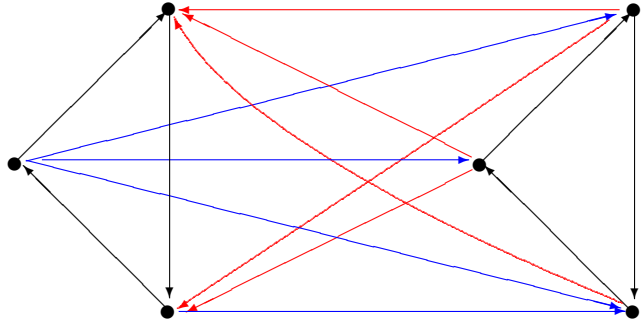}
\caption{An interesting tournament $\H$ on six vertices (self-loops are not drawn). This tournament does not retract to the $\DC^*_3$ on the left-hand side, yet $\mathrm{Spill}_3(\H[\DC^*_3,\DC_3]) = V(\H)$.}
\label{fig:weird-example}
\end{figure}

We now need to deal with the situation in which we have an isomorphic copy $\H'_0=i(\H_0)$ of $\H_0$ in $\H$ with $\mathrm{Spill}_m(\H[\H'_0,i(\HC_0)])=V(\H)$, such that $\H$ does not retract to $\H'_0$ (see Figure~\ref{fig:weird-example} for an example). We cannot deal with this case in a direct manner and first show another base case. For this we need the following lemma and an extension of endo-triviality that we discuss afterwards.

\begin{lemma}[\cite{STACStoTOCT}]
Let $\H$ be a reflexive tournament, containing a subtournament $\H_0$ so that any endomorphism of $\H$ that fixes $\H_0$ as a graph is an automorphism. Then any endomorphism of $\H$ that maps $\H_0$ to an isomorphic copy $\H'_0=i(\H_0)$ of itself is an automorphism of $\H$. 
\label{lem:7}
\end{lemma}

Let $\H_0$ be an induced subgraph of a digraph~$\H$. We say that the pair $(\H,\H_0)$ is \emph{endo-trivial} if all endomorphisms of $\H$ that fix $\H_0$ are automorphisms.

\begin{proposition}[Base Case II]
Let $\H$ be a reflexive tournament with a subtournament~$\H_0$ with Hamilton cycle $\HC_0$ so that $(\H,\H_0)$ and $\H_0$ are endo-trivial and $\mathrm{Spill}_m(\H[\H_0,\HC_0]) = V(\H)$. 
Then $\H$-{\sc Retraction} can be polynomially reduced to $\QCSP(\H)$.
\label{prop:main2}
\end{proposition}
\begin{proof}
Let $\G$ be an instance of {\sc $\H$-Retraction}. 
Let $m$ be the size of~$|V(\H_0)|$ and $n$ be the size of $|V(\H)|$. 
We build an instance $\phi$ of $\QCSP(\H)$
in the following fashion. Consider all possible functions $\lambda:[n]\rightarrow V(\H)$. For some such $\lambda$, let $\mathcal{G}(\lambda)$ be the graph enriched with constants $c_1,\ldots,c_n$ where these are interpreted over some subset of $V(\H)$ according to $\lambda$ in the natural way (acting on the subscripts).

Let $\mathcal{G}=\bigotimes_{\lambda \in {V(\H)}^{[n]}} \mathcal{G}(\lambda)$. Let $\G^d$, $\H^d$ and $\H^d_0$ be the diagonal copies of $\G$, $\H$ and $\H_0$ in $\mathcal{G}$. Let $\mathcal{H}$ be the subgraph of $\mathcal{G}$ induced by $V(\H) \times \cdots \times V(\H)$. Note that the constants $c_1,\ldots,c_n$ live in $\mathcal{H}$. Now build $\mathcal{G}'$ from $\mathcal{G}$ by augmenting a new copy 
of~$\Cylm$ for every vertex $v \in V(\mathcal{H}) \setminus V(\H^d_0)$. Vertex~$v$ is to be identified with any vertex in the top copy of~$\DCm$ in~$\Cylm$ and the bottom copy of~$\DCm$ is to be identified with $\HC_0$ in $\H^d_0$ according to the identity function. 

Finally, build $\phi$ from the canonical query of $\mathcal{G}'$ where we additionally turn the constants $c_1,\ldots,c_n$ to outermost universal variables.

First suppose that~$\G$ retracts to~$\H$ by $r$. Let $\lambda$ be some assignment of the universal variables of $\phi$ to $\H$. To prove $\phi \in \QCSP(\H)$ it suffices to prove that there is a homomorphism from $\mathcal{G}'$ to $\H$ that extends $\lambda$ and for this it suffices to prove that there is a homomorphism from $\mathcal{G}$ that extends $\lambda$. This is always possible since we have $\mathrm{Spill}_m(\H[\H_0,\HC_0])=V(\H)$ by assumption.  

Henceforth let us consider the homomorphic image of $\mathcal{G}$ that is $\mathcal{G}(\lambda)$. To prove $\phi \in \QCSP(\H)$ it suffices to prove that there is a homomorphism from $\G(\lambda)$ to $\H$ that extends $\lambda$. Note that it will be sufficient to prove that $\G$ retracts to $\H$. Well this was our original assumption so we are done.  

Suppose now $\phi \in \QCSP(\H)$. Choose some surjection for $\lambda$, the assignment of the universal variables of $\phi$ to $\H$. Recall $N=|V(\H)^{[n]}|$. The evaluation of the existential variables that witness $\phi \in \QCSP(\H)$ induces a surjective homomorphism $s$ from $\mathcal{G}'$ to $\H$ which contains within it a surjective homomorphism $s'$ from $\mathcal{H}=\H^N$ to $\H$. Consider the diagonal copy of $\H^d_0 \subset \H^d \subset \G^d$ in $(\G)^N$. By abuse of notation we will also consider each of $s$ and $s'$ acting just on the diagonal. If $|s'(\H^d_0)|=1$, by construction of $\mathcal{G}'$, we have $|s'(\H^d)|=1$. By Lemma  \ref{lem:uniformly-z}, this would mean $s'$ is uniformly mapping $\mathcal{H}$ to one vertex, which is impossible as $s'$ is surjective. Now we will work exclusively on the diagonal copy $\G^d$. As $1<|s'(\H^d_0)|<m$ is not possible either due to Lemma~\ref{lem:under-m}, we find that $|s'(\H^d_0)|=m$, and indeed $s'$ maps $\H^d_0$ to a copy of itself in $\H$ which we will call $\H'_0=i(\H^d_0)$ for some isomorphism~$i$. 

As $(\H,\H_0)$ is endo-trivial, Lemma~\ref{lem:7} tells us that  the restriction of $s'$ to $\H^d$ 
is an automorphism of $\H^d$, which we call $\alpha$. The required retraction from $\G$ to $\H$ is now given by $\alpha^{-1} \circ s'$.
\end{proof}

\subsection{The strongly connected case: Generalising the Base Cases}\label{s-general}

We now generalise the two base cases to more general cases via some recursive procedure. Afterwards we will show how to combine these two cases to complete our proof.
We will first need a slightly generalised version of Lemma~\ref{lem:7}, which nonetheless has virtually the same proof. \textcolor{black}{For completeness of this article we provide this proof from \cite{STACStoTOCT}.}

\begin{lemma}[\cite{STACStoTOCT}]
Let $\H_2 \supset \H_1 \supset H_0$ be a sequence of strongly connected reflexive tournaments, each one a subtournament of the one before. Suppose that any endomorphism of $\H_1$ that fixes $\H_0$ is an automorphism. Then any endomorphism $h$ of $\H_2$ that maps $\H_0$ to an isomorphic copy $\H'_0=i(\H_0)$ of itself also gives an isomorphic copy of $\H_1$ in $h(\H_1)$. 
\label{lem:8}
\end{lemma}
{\color{black}
\begin{proof}
For contradiction, suppose there is an endomorphism $h$ of $\H_2$ that maps $\H_0$ to an isomorphic copy $\H'_0=i(\H_0)$ of itself that does not yield an isomorphic copy of $\H_1$. In particular, $|h(\H_1)|<|V(\H_1)|$. We proceed as in the proof of the Lemma~\ref{lem:7}. Choose $h^{-1}$ in the following fashion. We let $h^{-1}$ of $h(\H_0)$ be the natural isomorphism of $h(\H_0)$ to $\H_0$ (that inverts the isomorphism given by $h$ from $\H_0$ to $\H'_0$). Otherwise we choose $h^{-1}$ arbitrarily, such that $h^{-1}(y)=x$ only if $h(x)=y$. Since 
$\H_2$ is a reflexive tournament, $h^{-1}$ is an isomorphism. And $h^{-1} \circ h$ is an endomorphism of $\H_2$ that fixes $\H_0$ that does not yield an isomorphic copy of $\H_1$ in $h(\H_1)$, a contradiction.
\end{proof}
}
The following two lemmas generalise Propositions~\ref{prop:main1} and \ref{prop:main2}.

\begin{proposition}[General Case I]
Let $\H_{0},\H_{1}, \ldots, \H_k, \H_{k+1}$
be reflexive tournaments, the first $k$ of which have Hamilton cycles $\HC_{0},\HC_{1}, \ldots, \HC_{k}$, respectively,
so that $\H_0 \subseteq H_1 \subseteq \cdots \subseteq \H_k \subseteq \H_{k+1}.$
Assume that $\H_0$, $(\H_1,\H_0)$, \ldots, $(\H_{k},\H_{k-1})$ are endo-trivial and that
\[
\begin{array}{lcl}
\mathrm{Spill}_{a_0}(\H_1[\H_0,\HC_{0}]) &= &V(\H_1) \\
\mathrm{Spill}_{a_1}(\H_2[\H_1,\HC_{1}]) &= &V(\H_2) \\
\hspace*{3mm} \vdots &\vdots &\hspace*{3mm}\vdots \\
\mathrm{Spill}_{a_{k-1}}(\H_{k}[\H_{k-1},\HC_{k-1}]) &= &V(\H_k).\\
\end{array}
\]
Moreover, assume that $\H_{k+1}$ retracts to $\H_k$ and also to every isomorphic copy $\H'_{k}=i(\H_k)$ of $\H_k$ in $\H_{k+1}$ with 
$\mathrm{Spill}_{a_k}(\H_{k+1}[\H'_k,i(\HC_{k})]) = V(\H_{k+1})$.
Then $\H_k$-{\sc Retraction} can be polynomially reduced to $\QCSP(\H_{k+1})$.
\label{prop:main-general1}
\end{proposition}
\begin{proof}
Let $a_{k+1}, \ldots, a_0$ be the cardinalities of $|V(\H_{k+1})|,\ldots,|V(\H_0|)$, respectively. Let $n=a_{k+1}$.
Let $\G$ be an instance of {\sc $\H_k$-Retraction}.
We will build an instance $\phi$ of $\QCSP(\H_{k+1})$ in the following fashion. First, take a copy of $\H_{k+1}$ together with $\G$ and build $\G'$ by identifying these on the copy of $\H_{k}$ that they both possess as  
an induced
subgraph.

Consider all possible functions $\lambda:[n]\rightarrow V(\H_{k+1})$. For some such $\lambda$, let $\mathcal{G'}(\lambda)$ be the graph enriched with constants $c_1,\ldots,c_n$ where these are interpreted over some subset of $V(\H_{k+1})$ according to $\lambda$ in the natural way (acting on the subscripts).

Let $\mathcal{G}'=\bigotimes_{\lambda \in {V(\H_{k+1})}^{[n]}} \mathcal{G}'(\lambda)$. Let $\G'^d$, $\H_{k+1}^d$ and $\H^d_{k}$ etc. be the diagonal copies of $\G'^d$, $\H_{k+1}$ and $\H_{k}$ in $\mathcal{G}'$. Let $\mathcal{H}_{k+1}$ be the subgraph of $\mathcal{G}'$ induced by $V(\H_{k+1}) \times \cdots \times V(\H_{k+1})$. Note that the constants $c_1,\ldots,c_n$ live in $\mathcal{H}_{k+1}$. Now build $\mathcal{G}''$ from $\mathcal{G}'$ by augmenting a new copy 
of~$\Cyl_{a_k}$ for every vertex $v \in V(\mathcal{H}_{k+1}) \setminus V(\H^d_{k})$. Vertex~$v$ is to be identified with any vertex in the top copy of~$\DC_{a_k}$ in~$\Cyl_{a_k}$ and the bottom copy of~$\DC_{a_k}$ is to be identified with $\HC_k$ in $\H^d_k$ according to the identity function. 

Then, for each $i \in [k]$, and $v \in V(\H^d_{i}) \setminus V(\H^d_{i-1})$, add a copy of $\Cyl_{a_{i-1}}$, where $v$ is identified with any vertex in the top copy of $\DC^*_{a_{i-1}}$ in $\Cyl_{a_{i-1}}$ and the bottom copy of $\DC^*_{i-1}$ is to be identified with $\H_{i-1}$ according to the identity map of $\DC^*_{a_{i-1}}$ to $\HC_{i-1}$.

Finally, build $\phi$ from the canonical query of $\mathcal{G}''$ where we additionally turn the constants $c_1,\ldots,c_n$ to outermost universal variables.

First suppose that~$\G$ retracts to~$\H_k$. Let $\lambda$ be some assignment of the universal variables of $\phi$ to $\H_{k+1}$. To prove $\phi \in \QCSP(\H_{k+1})$ it suffices to prove that there is a homomorphism from $\mathcal{G}''$ to $\H_{k+1}$ that extends $\lambda$ and for this it suffices to prove that there is a homomorphism from $\mathcal{G}'$ that extends $\lambda$. Let us explain why. We map the various copies of $\Cyl_{a_{i-1}}$ in $\G''$ in any suitable fashion, which will always exist due to our assumptions and the fact that $\mathrm{Spill}_{a_{k}}(\H_{k+1}[\H_{k},\HC_{k}]) = V(\H_{k+1})$, which follows from our assumption that $\H_{k+1}$ retracts to $\H_k$ and Lemma~\ref{lem:spillage}.

Henceforth let us consider the homomorphic image of $\mathcal{G}'$ that is $\mathcal{G}'(\lambda)$. To prove $\phi \in \QCSP(\H_{k+1})$ it suffices to prove that there is a homomorphism from $\G'(\lambda)$ to $\H_{k+1}$ that extends $\lambda$. Note that it will be sufficient to prove that $\G'$ retracts to $\H_{k+1}$. Let $h$ be the natural retraction from~$\G'$ to~$\H_{k+1}$ that extends the known retraction from $\G$ to $\H_k$. We are done.  

Suppose now $\phi \in \QCSP(\H_{k+1})$. Choose some surjection for $\lambda$, the assignment of the universal variables of $\phi$ to $\H_{k+1}$. Let $N=|V(\H_{k+1})^{[n]}|$. The evaluation of the existential variables that witness $\phi \in \QCSP(\H_{k+1})$ induces a surjective homomorphism $s$ from $\mathcal{G}'$ to $\H_{k+1}$ which contains within it a surjective homomorphism $s'$ from $\mathcal{H}=\H_{k+1}^N$ to $\H_{k+1}$. Consider the diagonal copy of $\H^d_0 \subset \cdots \subset \H^d_k \subset \H^d_{k+1} \subset G'^d$ in $\mathcal{G}'$. By abuse of notation we will also consider each of $s$ and $s'$ acting just on the diagonal. If $|s'(\H^d_0)|=1$, by construction of $\mathcal{G}''$, we could follow the chain of spills to deduce that $|s'(\H^d_{k+1})|=1$, which is not possible by Lemma  \ref{lem:uniformly-z}. Moreover, $1<|s'(H_0^d)|<|V(H_0^d)|$ is impossible due to Lemma \ref{lem:under-m}. Now we will work exclusively on the diagonal copy $\G'^d$.

Thus, $|s'(\H^d_0)|=|V(\H^d_0)|$ and indeed $s'$ maps $\H^d_0$ to an isomorphic copy of itself in $\H_{k+1}$ which we will call $\H'_0=i(\H^d_0)$. We now apply Lemma~\ref{lem:8} as well as our assumed endo-trivialities to derive that $s'$ in fact maps $\H^d_k$ by the isomorphism $i$ to a copy of itself in $\H_{k+1}$ which we will call $\H'_{k}$. Since $s'$ is surjective, we can deduce that $\mathrm{Spill}_{a_k}(\H_{k+1}[\H'_k,i(\HC^d_{k})]) = V(\H_{k+1})$ in the same way as in the proof of Proposition~\ref{prop:main1}.
 and so there exists a retraction $r$ from $\H_{k+1}$ to $\H'_k$.  Now $i^{-1} \circ r \circ s'$ gives the desired retraction of $\G$ to $\H_k$.   
\end{proof}

\begin{proposition}[General Case II]
\label{prop:main-general2}
Let  $\H_{0},\H_{1}, \ldots, \H_k, \H_{k+1}$ be reflexive tournaments, the first $k+1$ of which have Hamilton cycles $\HC_{0},\HC_{1}, \ldots, \HC_{k}$, respectively,
so that $\H_0 \subseteq H_1 \subseteq \cdots \subseteq \H_k \subseteq \H_{k+1}$.
Suppose that $\H_0$, $(\H_1,\H_0)$, \ldots, $(\H_{k},\H_{k-1}), (\H_{k+1},\H_{k})$ are endo-trivial and that
\[
\begin{array}{lcl}
\mathrm{Spill}_{a_0}(\H_1[\H_0,\HC_{0}]) &= &V(\H_1) \\
\mathrm{Spill}_{a_1}(\H_2[\H_1,\HC_{1}]) &= &V(\H_2) \\
\hspace*{3mm} \vdots &\vdots &\hspace*{3mm}\vdots \\
\mathrm{Spill}_{a_{k-1}}(\H_{k}[\H_{k-1},\HC_{k-1}]) &= &V(\H_k)\\
\mathrm{Spill}_{a_{k}}(\H_{k+1}[\H_{k},\HC_{k}]) &= &V(\H_{k+1})
\end{array}
\]
Then {\sc $\H_{k+1}$-Retraction} can be polynomially reduced to $\QCSP(\H_{k+1})$.
\end{proposition}
\begin{proof}
Let $n=a_{k+1}=|V(\H_{k+1})|$ and let $a_k, \ldots, a_0$ be the cardinalities of $|V(\H_k)|,\ldots,|V(\H_0)|$, respectively.
Let $\G$ be an instance of {\sc $\H_{k+1}$-Retraction}.
We build an instance  $\phi$ of $\QCSP(\H_{k+1})$
in the following fashion. Consider all possible functions $\lambda:[n]\rightarrow V(\H_{k+1})$. For some such $\lambda$, let $\mathcal{G}(\lambda)$ be the graph enriched with constants $c_1,\ldots,c_n$ where these are interpreted over some subset of $V(\H_{k+1})$ according to $\lambda$ in the natural way (acting on the subscripts).

Let $\mathcal{G}=\bigotimes_{\lambda \in {V(\H_{k+1})}^{[n]}} \mathcal{G}(\lambda)$. Let $\G^d$, $\H^d_{k+1}$,  $\H^d_k$, \ldots, $\H^d_0$ be the diagonal copies of $\G$, $\H_{k+1},\H_k,\ldots,\H_0$ in $\mathcal{G}$. Let $\mathcal{H}_{k+1}$ be the subgraph of $\mathcal{G}$ induced by $V(\H_{k+1}) \times \cdots \times V(\H_{k+1})$. Note that the constants $c_1,\ldots,c_n$ live in $\mathcal{H}_{k+1}$.

Build $\mathcal{G}'$ from $\mathcal{G}$ by first augmenting a new copy 
of~$\Cyl_{a_k}$ for every vertex $v \in V(\mathcal{H}_{k+1}) \setminus V(\H^d_k)$. Vertex~$v$ is to be identified with any vertex in the top copy of~$\DC_{a_k}$ in~$\Cyl_{a_k}$ and the bottom copy of~$\DC_{a_k}$ is to be identified with $\HC_k$ in $\H^d_k$ according to the identity function. Now, for each $i \in [k]$, and $v \in V(\H^d_{i}) \setminus V(\H^d_{i-1})$, we add a copy of $\Cyl_{a_{i-1}}$, where $v$ is identified with any vertex in the top copy of $\DC^*_{a_{i-1}}$ in $\Cyl_{a_{i-1}}$ and the bottom copy of $\DC^*_{i-1}$ is to be identified with $\H^d_{i-1}$ according to the identity map of $\DC^*_{a_{i-1}}$ to $\HC^d_{i-1}$.

Finally, build $\phi$ from the canonical query of $\mathcal{G}'$ where we additionally turn the constants $c_1,\ldots,c_n$ to outermost universal variables.

First suppose that $\G$ retracts to $\H_{k+1}$. Let $h$ be a retraction from $\G$ to $\H_{k+1}$. Let $\lambda$ be some assignment of the universal variables of $\phi$ to $\H_{k+1}$. To prove $\phi \in \QCSP(\H_{k+1})$ it suffices to prove that there is a homomorphism from $\mathcal{G}'$ to $\H_{k+1}$ that extends $\lambda$ and for this it suffices to prove that there is a homomorphism from $\mathcal{G}$ that extends $\lambda$.  The extension of the latter to the former will always be possible due to the spill assumptions.

Henceforth let us consider the homomorphic image of $\mathcal{G}$ that is $\mathcal{G}(\lambda)$. To prove $\phi \in \QCSP(\H_{k+1})$ it suffices to prove that there is a homomorphism from $\mathcal{G}(\lambda)$ to $\H_{k+1}$ that extends $\lambda$. Note that it will be sufficient to prove that $\G$ retracts to $\H_{k+1}$. Well this was our original assumption so we are done. 

Suppose now $\phi \in \QCSP(\H_{k+1})$. Choose some surjection for $\lambda$, the assignment of the universal variables of $\phi$ to $\H_{k+1}$. Let $N=|V(\H_{k+1})^{[n]}|$. The evaluation of the existential variables that witness $\phi \in \QCSP(\H_{k+1})$ induces a surjective homomorphism $s$ from $\mathcal{G}$ to $\H_{k+1}$ which contains within it a surjective homomorphism $s'$ from $\mathcal{H}_{k+1}=\H_{k+1}^N$ to $\H_{k+1}$. Consider the diagonal copy of $\H^d_0 \subset \H^d_1 \subset \cdots H^d_{k+1}$ in $\mathcal{G}$. By abuse of notation we will also consider each of $s$ and $s'$ acting just on the diagonal. If $|s'(\H^d_0)|=1$, by construction of $\mathcal{G}'$, we have $|s'(\H^d)|=1$. Now we follow the chain of spills to deduce that $|s'(\mathcal{H}_{k+1})|=1$, a contradiction. We now apply Lemma~\ref{lem:8} as well as our assumed endo-trivialities to derive that $s'$ in fact maps $\H^d_k$ by the isomorphism $i$ to a copy of itself in $\H_{k+1}$, which we will call $\H'_{k}$. Now we can deduce, via Lemma~\ref{lem:7}, that $s'(\H^d_{k+1})$ is an automorphism of $\H_{k+1}$, which we call $\alpha$.
The required retraction from $\G$ to 
$\H_{k+1}$ 
is now given by $\alpha^{-1} \circ s'$.
\end{proof}

\begin{corollary}
Let $\H$ be a non-trivial strongly connected reflexive tournament. 
Then $\QCSP(\H)$ is \NP-hard.
\label{cor:strongly-connected1}
\end{corollary}
\begin{proof}
As $\H$ is a strongly connected reflexive tournament, which has more than one vertex by our assumption, $\H$ is not transitive.
Note that {\sc $\H$-Retraction}
 is \NP-complete (see Section 4.5 in \cite{STACStoTOCT}, using results from \cite{La06,JBK05,LZ05}). Thus, if $\H$ is endo-trivial, the result follows from Proposition~\ref{prop:main1} (note that we could also have used Corollary~\ref{c-main}).

Suppose $\H$ is not endo-trivial. Then, by Lemma~\ref{lem:endo-retraction-trivial}, $\H$ is not retract-trivial either. This means
that~$\H$ has a non-trivial retraction to some subtournament~$\H_0$. We may assume
that $\H_0$ is endo-trivial, as otherwise we will repeat the argument until we find a retraction from $\H$ to an endo-trivial (and consequently strongly connected) subtournament.
 
Suppose that $\H$ retracts to all isomorphic copies $\H'_0=i(\H_0)$ of $\H_0$ within it, except possibly those for which $\mathrm{Spill}_m(\H[\H'_0,i(\HC_0)]) \neq V(\H)$. Then the result follows from  Proposition~\ref{prop:main1}. So there is a copy $\H'_0=i(\H_0)$ to which $\H$ does not retract for which $\mathrm{Spill}_m(\H[\H'_0,i(\HC_0)]) = V(\H)$.
If $(\H,\H'_0)$ is endo-trivial, the result follows from Proposition~\ref{prop:main2}. Thus we assume $(\H,\H'_0)$ is not endo-trivial and we deduce the existence of $\H'_0 \subset \H_1 \subset \H$ ($\H_1$ is strictly between $\H$ and $\H'_0$) so that $(\H_1,\H'_0)$ and $H'_0$ are endo-trivial and $\H$ retracts to $\H_1$. Now we are ready to break out. Either $\H$ retracts to all isomorphic copies of $\H'_1=i(\H_1)$ in $\H$, except possibly for those so that $\mathrm{Spill}_m(\H[\H'_1,i(\HC_1)]) \neq V(\H)$, and we apply Proposition~\ref{prop:main-general1}, or there exists a copy $\H'_1$, with $\mathrm{Spill}_m(\H[\H'_1,i(\HC_1)]) = V(\H)$, to which it does not retract. 
If $(\H,\H'_1)$ is endo-trivial, the result follows from Proposition~\ref{prop:main-general2}. Otherwise we iterate the method, which will terminate because our structures are getting strictly \textcolor{black}{smaller}. 
\end{proof}

\subsection{An initial strongly connected component that is non-trivial}
\label{sec-connectedc:initial-strongly}

Let $\H^+$ denote any reflexive tournament that has an initial strongly connected component $\H$ that is non-trivial (not of size $1$). Let $\Cylmplus$ be $\Cylm$ but with a pendant out-edge hanging from the top-most cycle. This edge is directed to the vertex $x$. Thus, $\Cylmplus$ contains one additional vertex to $\Cylm$ and this has an incoming edge from some vertex in the top-most cycle $\DCm$ (it does not matter which one). $\Cylmplus$ is drawn in Figure~\ref{fig:Photo2plus}.

Define $\mathrm{Spill}^+_m$ as $\mathrm{Spill}_m$ but with respect to $\Cylmplus$ instead of $\Cylm$. 
{\color{black} At this point we risk confusion with our overburdened notation. Let us address in Table~\ref{table-confusion} how our notation maps from the strongly connected case to that in which there is an initial strongly connected component that is non-trivial.}

\begin{table}
{\color{black}
\begin{tabular}{|c|c|c|}
\hline
& Strongly connected & \begin{tabular}{c}Initial component \\ strongly connected \end{tabular} \\
\hline
Graph & $\H$ & $\H^+$ \\
\hline
Gadget & $\Cylm$ & $\Cylmplus$ \\
\hline
\begin{tabular}{c}Subgraph \\ (strongly connected) \end{tabular} & $\H_0$ & $\H_0$ \\
\hline
Hamilton cycle & $\HC_0$ & $\HC_0$ \\
\hline
Spill & $\mathrm{Spill}_m(\H[\H_0,\HC_0))$ & $\mathrm{Spill}^+_m(\H^+[\H_0,\HC_0])$ \\
\hline
\end{tabular}
}
\caption{\textcolor{black}{Mapping notation from the strongly connected case to the case in which there is an initial strongly connected component that is non-trivial.}}
\label{table-confusion}
\end{table}
 \begin{figure}
\centering
\input{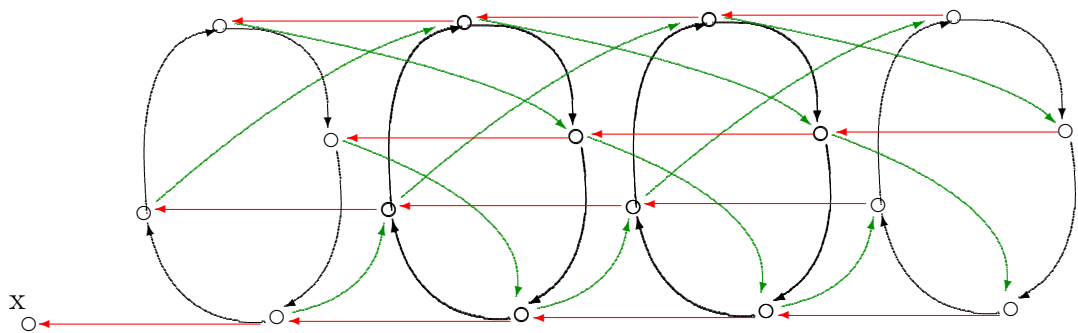}
\caption{The gadget $\Cylmplus$ in the case $m:=4$ (self-loops are not drawn). We usually visualise the right-hand copy of~$\DC^*_4$ as the ``bottom'' copy and then we talk about vertices ``above'' and ``below'' according to the red arrows. The vertex $x$ is depicted at the left-hand extremity.}
\label{fig:Photo2plus}
\end{figure}
Note that Lemma~\ref{lem:claim1}, with $\Cylm$ replaced by $\Cylmplus$, does not hold.

\begin{lemma}
Let $\H^+$ be some reflexive tournament that has an initial strongly connected component $\H$ that is non-trivial and contains endo-trival $\H_0$ with Hamilton cycle $\HC_0$. Suppose $\mathrm{Spill}^+_m(\H[\H_0,\HC_0])=V(\H)$, then $\mathrm{Spill}^+_m(\H^+[\H_0,\HC_0])=V(\H^+)$.
\label{lem:new-trivial}
\end{lemma}
\begin{proof}
We only need to argue for the $x \in \H^+\setminus H$. In this case, we may evaluate all the cycles in $\Cylmplus$ onto $\HC_0$ with each vertex mapping to the one directly beneath it. This works as $x$ is forward-adjacent from every vertex in $\HC_0$.
\end{proof}
\noindent The condition of endo-triviality of $\H_0$ was not used in the proof of Lemma~\ref{lem:new-trivial}.
\begin{proposition}[Base Case A-I.]
Let $\H^+$ be some reflexive tournament that has an initial strongly connected component $\H$ that is non-trivial and contains endo-trivial $\H_0$ with Hamilton cycle $\HC_0$. 
Assume that $\H$ retracts to $\H'_0$ for every isomorphic copy $\H'_0=i(\H_0)$ of~$\H_0$ in~$\H$ with $\mathrm{Spill}^+_m(\H[\H'_0,i(\HC_0)])=V(\H)$.
Then $\H_0$-{\sc Retraction} can be polynomially reduced to $\QCSP(\H^+)$.
\label{prop:mainA1}
\end{proposition}

\begin{proof}
Let $m$ be the size of~$|V(\H_0)|$ and $n$ be the size of $|V(\H)|$. 
Let~$\G$ be an instance of $\H_0$-{\sc Retraction}.
We build an instance~$\phi$ of $\QCSP(\H^+)$
in the following fashion. First, take a copy of~$\H$ together with~$\G$ and build~$\G'$ by identifying these on the copy of~$\H_0$ that they both possess as an induced subgraph.

Now, consider all possible functions $\lambda:[n]\rightarrow V(\H)$. For some such $\lambda$, let $\mathcal{G'}(\lambda)$ be the graph enriched with constants $c_1,\ldots,c_n$ where these are interpreted over some subset of $V(\H)$ according to $\lambda$ in the natural way (acting on the subscripts).

Let $\mathcal{G}'=\bigotimes_{\lambda \in {V(\H)}^{[n]}} \mathcal{G}'(\lambda)$. Let $\G'^d$, $\H^d$ and $\H^d_0$ be the diagonal copies of $\G'$, $\H$ and $\H_0$ in $\mathcal{G}'$. Let $\mathcal{H}$ be the subgraph of $\mathcal{G}'$ induced by $V(\H) \times \cdots \times V(\H)$. Note that the constants $c_1,\ldots,c_n$ live in $\mathcal{H}$. Now build $\mathcal{G}''$ from $\mathcal{G}'$ by augmenting a new copy of~$\Cylmplus$ for every vertex $v \in V(\mathcal{H}) \setminus V(\H^d_0)$. Vertex~$v$ is to be identified with the vertex $x$ that is at the end of the out-edge pendant on the top copy of~$\DCm$ in~$\Cylmplus$ and the bottom copy of~$\DCm$ is to be identified with $\HC_0$ in $\H^d_0$ according to the identity function. Call these \emph{the $\Cylmplus$ of the second stage}.

Now build $\mathcal{G}'''$ by adding an edge from each vertex $c_i$ to a new vertex $d_i$ (for each $i \in [n]$). Now add a copy 
of~$\Cylmplus$ for every vertex $v \in \{d_1,\ldots,d_n\}$. Vertex~$v$ is to be identified with the vertex $x$ that is at the end of the out-edge pendant on the top copy of~$\DCm$ in~$\Cylmplus$ and the bottom copy of~$\DCm$ is to be identified with $\HC_0$ in $\H^d_0$ according to the identity function.  Call these \emph{the $\Cylmplus$ of the third stage}.

Finally, build $\phi$ from the canonical query of $\mathcal{G}'''$, where we additionally turn the vertices $d_1,\ldots,d_n$ to outermost universal variables $z_1,\ldots,z_n$. Then existentially quantify all remaining constants and vertices innermost. Finally, restrict all except the universal variables to be in $V(\H)$, appealing to the definition guaranteed by Corollary~\ref{cor:galois}.


We claim that~$\G$ retracts to~$\H_0$ if and only if $\phi \in \QCSP(\H^+)$.

First suppose that~$\G$ retracts to~$\H_0$ by $r$. Let $\lambda'$ be some assignment of the universal variables $z_1,\ldots,z_n$ of $\phi$ to $\H^+$ and choose $y_1,\ldots,y_n$ backwards-adjacent to these in $\H$, mapped by $\lambda$. To prove $\phi \in \QCSP(\H^+)$ it suffices to prove that there is a homomorphism from $\mathcal{G}''$ to $\H^+$ that extends $\lambda$ and for this it suffices to prove that there is a homomorphism $h$ from $\mathcal{G}'$ to $\H$ that extends $\lambda$. Let us explain why. Because $\H$ retracts to $\H_0$, we have $\mathrm{Spill}_m(\H[\H_0,\HC_0])=V(\H)$ due to Lemma~\ref{lem:spillage} which implies the weaker $\mathrm{Spill}^+_m(\H[\H_0,\HC_0])=V(\H)$. For the $\Cylmplus$ of the second stage, the weaker statement suffices, but for the $\Cylmplus$ of the third stage, the stronger statement is needed. 

Henceforth let us consider the homomorphic image of $\mathcal{G}'$ that is $\mathcal{G}'(\lambda)$. To prove $\phi \in \QCSP(\H^+)$ it suffices to prove that there is a homomorphism from $\G'(\lambda)$ to $\H$ that extends $\lambda$. Note that it will be sufficent to prove that $\G'$ retracts to $\H$. Let $h$ be the natural retraction from~$\G'$ to~$\H$ that extends the known retraction $r$ from $\G$ to $\H_0$. We are done.  

Suppose now $\phi \in \QCSP(\H^+)$.  Choose some surjection for $\lambda'$ mapping $z_1,\ldots,z_n$ to $\H$. Choose some $y_1,\ldots,y_n$ backwards-adjacent to these and let this be the map $\lambda$. Note that it is not possible for all $y_1,\ldots,y_n$ to be evaluated as a single vertex as the initial strongly connected component is non-trivial.

The evaluation of the existential variables that witness $\phi \in \QCSP(\H)$ induces a non-trivial homomorphism $s$ from $\mathcal{G}''$ to $\H$ which contains within it a non-trivial homomorphism $s'$ from $\mathcal{H}=\H^N$ to $\H$. Consider the diagonal copy of $\H^d_0 \subset \H^d \subset \G'^d$ in $\mathcal{G}'$. By abuse of notation we will also consider each of $s$ and $s'$ acting just on the diagonal.

 If $|s'(\H^d_0)|=1$, by construction of $\mathcal{G}''$, we have that $s'(\H^d)$ is an in-star (that is, a single terminal vertex receiving an edge from potentially numerous initial vertices), but this is not possible as $\H^d$ is strongly connected. As $1<|s'(\H^d_0)|<m$ is not possible either due to Lemma~\ref{lem:under-m}, we find that $|s'(\H^d_0)|=m$, and indeed $s'$ maps $\H^d_0$ to a copy of itself in $\H$ which we will call $\H'_0=i(\H^d_0)$ for some isomorphism~$i$.

We claim that $\mathrm{Spill}^+_m(\H[\H'_0,i(\HC^d_0)]) = V(\H)$. Since $\lambda'$ is surjective on $\H^+$, this is enforced explicitly by the $\Cylmplus$ of the third stage. As $\mathrm{Spill}^+_m(\H[\H'_0,i(\HC^d_0)]) = V(\H)$, we find, by assumption of the lemma, that there exists a retraction $r$ from $\H^d$ to $\H'_0$.  
Now $i^{-1} \circ r \circ s'$ is the desired retraction of $\G$ to~$\H_0$.  
\end{proof}

\begin{proposition}[Base Case A-II]
Let $\H^+$ be some reflexive tournament that has an initial strongly connected component $\H$ that is non-trivial and contains $\H_0$ with Hamilton cycle $\HC_0$ so that  $(\H,\H_0)$ and $\H_0$ are endo-trivial and $\mathrm{Spill}^+_m(\H[\H_0,\HC_0]) = V(\H)$. 
Then $\H$-{\sc Retraction} can be polynomially reduced to $\QCSP(\H^+)$.
\label{prop:mainA2}
\end{proposition}
\begin{proof}
Let $m$ be the size of~$|V(\H_0)|$ and $n$ be the size of $|V(\H)|$.
Let $\G$ be an instance of {\sc $\H$-Retraction}. 
We build an instance $\phi$ of $\QCSP(\H^+)$
in the following fashion. Consider all possible functions $\lambda:[n]\rightarrow V(\H)$. For some such $\lambda$, let $\mathcal{G}(\lambda)$ be the graph enriched with constants $c_1,\ldots,c_n$ where these are interpreted over some subset of $V(\H)$ according to $\lambda$ in the natural way (acting on the subscripts).

Let $\mathcal{G}=\bigotimes_{\lambda \in {V(\H)}^{[n]}} \mathcal{G}(\lambda)$. Let $\G^d$, $\H^d$ and $\H^d_0$ be the diagonal copies of $\G$, $\H$ and $\H_0$ in $\mathcal{G}$. Let $\mathcal{H}$ be the subgraph of $\mathcal{G}$ induced by $V(\H) \times \cdots \times V(\H)$. Note that the constants $c_1,\ldots,c_n$ live in $\mathcal{H}$. Now build $\mathcal{G}'$ from $\mathcal{G}$ by augmenting a new copy 
of~$\Cylmplus$ for every vertex $v \in V(\mathcal{H}) \setminus V(\H^d_0)$. Vertex~$v$ is to be identified with the vertex $x$ that is at the end of the out-edge pendant on the top copy of~$\DCm$ in~$\Cylmplus$ and the bottom copy of~$\DCm$ is to be identified with $\HC_0$ in $\H^d_0$ according to the identity function. 

Now build $\mathcal{G}''$ by adding an edge from each vertex $c_i$ to a new vertex $d_i$ (for each $i \in [n]$). 

Finally, build $\phi$ from the canonical query of $\mathcal{G}''$, where we additionally turn the vertices $d_1,\ldots,d_n$ to outermost universal variables $z_1,\ldots,z_n$. Then existentially quantify all remaining constants and vertices innermost. Finally, restrict all except the universal variables to be in $V(\H)$.

First suppose that~$\G$ retracts to~$\H$ by $r$. Let $\lambda'$ be some assignment of the universal variables $z_1,\ldots,z_n$ of $\phi$ to $\H^+$ and choose $y_1,\ldots,y_n$ backwards-adjacent to these in $\H$, mapped by $\lambda$.

To prove $\phi \in \QCSP(\H^+)$ it suffices to prove that there is a homomorphism from $\mathcal{G}'$ to $\H^+$ that extends $\lambda$ and for this it suffices to prove that there is a homomorphism $h$ from $\mathcal{G}$ to $\H$ that extends $\lambda$. Let us explain why. 
By assumption, we have $\mathrm{Spill}^+_m(\H[\H_0,\HC_0])=V(\H)$. 

Henceforth let us consider the homomorphic image of $\mathcal{G}$ that is $\mathcal{G}(\lambda)$. To prove $\phi \in \QCSP(\H^+)$ it suffices to prove that there is a homomorphism from $\G(\lambda)$ to $\H$ that extends $\lambda$. Note that it will be sufficient to prove that $\G$ retracts to $\H$. 
We are done.  

Suppose now $\phi \in \QCSP(\H^+)$. Choose some surjection for $\lambda'$ mapping $z_1,\ldots,z_n$ to $\H$. Choose some $y_1,\ldots,y_n$ backwards-adjacent to these (and therefore in $\H$) and let this be the map $\lambda$. Note that it is not possible for all $y_1,\ldots,y_n$ to be evaluated as a single vertex as $\H$ is strongly connected. Recall $N=|V(\H)^{[n]}|$. The evaluation of the existential variables that witness $\phi \in \QCSP(\H)$ induces a non-trivial homomorphism $s$ from $\mathcal{G}'$ to $\H$ which contains within it a non-trivial homomorphism $s'$ from $\mathcal{H}=\H^N$ to $\H$. Consider the diagonal copy of $\H^d_0 \subset \H^d \subset \G^d$ in $\mathcal{G}$. By abuse of notation we will also consider each of $s$ and $s'$ acting just on the diagonal. 
If $|s'(\H^d_0)|=1$, by construction of $\mathcal{G}''$ with the $\Cylmplus$, we have $s'(\H^d)$ is an in-star, but this is not possible as $\H^d$ is strongly connected. As $1<|s'(\H^d_0)|<m$ is not possible either due to Lemma~\ref{lem:under-m}, we find that $|s'(\H^d_0)|=m$, and indeed $s'$ maps $\H^d_0$ to a copy of itself in $\H$ which we will call $\H'_0=i(\H^d_0)$ for some isomorphism~$i$. 

As $(\H,\H_0)$ is endo-trivial, Lemma~\ref{lem:7} tells us that  the restriction of $s'$ to $\H^d$ 
is an automorphism of $\H^d$, which we call $\alpha$. The required retraction from $\G$ to $\H$ is now given by $\alpha^{-1} \circ s'$.
\end{proof}

It remains to generalise these base cases.

\begin{proposition}[General Case A-I]
Let $\H_{k+1}^+$ be some reflexive tournament that has an initial strongly connected component $\H_{k+1}$.
Let $\H_{0},\H_{1}, \ldots, \H_k, \H_{k+1}$
be reflexive tournaments, the first $k$ of which have Hamilton cycles $\HC_{0},\HC_{1}, \ldots, \HC_{k}$, respectively,
so that $\H_0 \subseteq H_1 \subseteq \cdots \subseteq \H_k \subseteq \H_{k+1}.$
Assume that $\H_0$, $(\H_1,\H_0)$, \ldots, $(\H_{k},\H_{k-1})$ are endo-trivial and that
\[
\begin{array}{lcl}
\mathrm{Spill}^+_{a_0}(\H_1[\H_0,\HC_{0}]) &= &V(\H_1) \\
\mathrm{Spill}^+_{a_1}(\H_2[\H_1,\HC_{1}]) &= &V(\H_2) \\
\hspace*{3mm} \vdots &\vdots &\hspace*{3mm}\vdots \\
\mathrm{Spill}^+_{a_{k-1}}(\H_{k}[\H_{k-1},\HC_{k-1}]) &= &V(\H_k).\\
\end{array}
\]
Moreover, assume that $\H_{k+1}$ retracts to $\H_k$ and also to every isomorphic copy $\H'_{k}=i(\H_k)$ of $\H_k$ in $\H_{k+1}$ with 
$\mathrm{Spill}^+_{a_k}(\H_{k+1}[\H'_k,i(\HC_{k})]) = V(\H_{k+1})$.
Then $\H_k$-{\sc Retraction} can be polynomially reduced to $\QCSP(\H^+_{k+1})$.
\label{prop:main-generala-1}
\end{proposition}
\begin{proof}
Let $n=a_{k+1}=|V(\H_{k+1})|$ and let $a_k, \ldots, a_0$ be the cardinalities of $|V(\H_k)|,\ldots,|V(\H_0)|$, respectively.
Let $\G$ be an instance of {\sc $\H_k$-Retraction}.
We will build an instance $\phi$ of $\QCSP(\H^+_{k+1})$ in the following fashion. First, take a copy of $\H_{k+1}$ together with $\G$ and build $\G'$ by identifying these on the copy of $\H_{k}$ that they both possess as 
an induced
subgraph.

Consider all possible functions $\lambda:[n]\rightarrow V(\H_{k+1})$. For some such $\lambda$, let $\mathcal{G'}(\lambda)$ be the graph enriched with constants $c_1,\ldots,c_n$ where these are interpreted over some subset of $V(\H_{k+1})$ according to $\lambda$ in the natural way (acting on the subscripts).

Let $\mathcal{G}'=\bigotimes_{\lambda \in {V(\H_{k+1})}^{[n]}} \mathcal{G}'(\lambda)$. Let $\G'^d$, $\H_{k+1}^d$ and $\H^d_{k}$ etc. be the diagonal copies of $\G'$, $\H_{k+1}$ and $\H_{k}$ in $\mathcal{G}'$. Let $\mathcal{H}_{k+1}$ be the subgraph of $\mathcal{G}'$ induced by $V(\H_{k+1}) \times \cdots \times V(\H_{k+1})$. Note that the constants $c_1,\ldots,c_n$ live in $\mathcal{H}_{k+1}$. 

Now build $\mathcal{G}''$ from $\mathcal{G}'$ by augmenting a new copy 
of~$\Cylplus_{a_k}$ for every vertex $v \in V(\mathcal{H}_{k+1}) \setminus V(\H^d_{k})$. Vertex~$v$ is to be identified with the vertex $x$ that is at the end of the out-edge pendant on the top copy of~$\DC_{a_k}$ in~$\Cylplus_{a_k}$ and the bottom copy of~$\DC_{a_k}$ is to be identified with $\HC_k$ in $\H^d_k$ according to the identity function.  Call these \emph{the $\Cylplus_{a_k}$  of the second stage}. Then, for each $i \in [k]$, and $v \in V(\H^d_{i}) \setminus V(\H^d_{i-1})$, add a copy of $\Cylplus_{a_{i-1}}$, where $v$ is identified with the vertex $x$ that is at the end of the out-edge pendant on the top copy of $\DC^*_{a_{i-1}}$ in $\Cylplus_{a_{i-1}}$ and the bottom copy of $\DC^*_{i-1}$ is to be identified with $\H_{i-1}$ according to the identity map of $\DC^*_{a_{i-1}}$ to $\HC_{i-1}$.

Now build $\mathcal{G}'''$ by adding an edge from each vertex $c_i$ to a new vertex $d_i$ (for each $i \in [n]$). Now add a copy 
of~$\Cylplus_{a_k}$ for every vertex $v \in \{d_1,\ldots,d_n\}$. Vertex~$v$ is to be identified with the vertex $x$ that is at the end of the out-edge pendant on the top copy of~$\DC_{a_k}$ in~$\Cylplus_{a_k}$  and the bottom copy of~$\DC_{a_k}$ is to be identified with $\HC_k$ in $\H^d_k$ according to the identity function.  Call these \emph{the $\Cylplus_{a_k}$  of the third stage}.

Finally, build $\phi$ from the canonical query of $\mathcal{G}'''$, where we additionally turn the vertices $d_1,\ldots,d_n$ to outermost universal variables $z_1,\ldots,z_n$. Then existentially quantify all remaining constants and vertices innermost. Finally, restrict all except the universal variables to be in $V(\H)$.

First suppose that~$\G$ retracts to~$\H_k$ by $r$. Let $\lambda'$ be some assignment of the universal variables $z_1,\ldots,z_n$ of $\phi$ to $\H_{k+1}^+$ and choose $y_1,\ldots,y_n$ backwards-adjacent to these in $\H_{k+1}$, mapped by $\lambda$. To prove $\phi \in \QCSP(\H_{k+1}^+)$ it suffices to prove that there is a homomorphism from $\mathcal{G}''$ to $\H^+_{k+1}$ that extends $\lambda$ and for this it suffices to prove that there is a homomorphism $h$ from $\mathcal{G}'$ that extends $\lambda$. Let us explain why. Because $\H_{k+1}$ retracts to $\H_k$, we have $\mathrm{Spill}_{a_{k}}(\H_{k+1}[\H_k,\HC_k])=V(\H_{k+1})$ due to Lemma~\ref{lem:spillage} which implies the weaker $\mathrm{Spill}^+_{a_{k}}(\H_{k+1}[\H_k,\HC_k])=V(\H_{k+1})$. For the $\Cylplus_{a_{k}}$ of the second stage, the weaker statement suffices, but for the $\Cylplus_{a_{k}}$ of the third stage, the stronger statement is needed. We continue mapping now the various copies of $\Cylplus_{a_{i-1}}$ in $\G''$ in any suitable fashion, which will always exist due to our assumptions.

Henceforth let us consider the homomorphic image of $\mathcal{G}'$ that is $\mathcal{G}'(\lambda)$. To prove $\phi \in \QCSP(\H^+_{k+1})$ it suffices to prove that there is a homomorphism from $\G'(\lambda)$ to $\H_{k+1}$ that extends $\lambda$. Note that it will be sufficient to prove that $\G'$ retracts to $\H_{k+1}$. Let $h$ be the natural retraction from~$\G'$ to~$\H_{k+1}$ that extends the known retraction $r$ from $\G$ to $\H_k$. We are done.  

Suppose now $\phi \in \QCSP(\H^+_{k+1})$. Choose some surjection for $\lambda$, the assignment of the universal variables of $\phi$ to $\H_{k+1}$. Choose some $y_1,\ldots,y_n$ backwards-adjacent to these (and therefore in $\H_{k+1}$) and let this be the map $\lambda$. Note that it is not possible for all $y_1,\ldots,y_n$ to be evaluated as a single vertex as $\H_{k+1}$ is strongly connected. Let $N=|V(\H_{k+1})^{[n]}|$. The evaluation of the existential variables that witness $\phi \in \QCSP(\H^+_{k+1})$ induces a non-trivial homomorphism $s$ from $\mathcal{G}'$ to $\H_{k+1}$ which contains within it a non-trivial homomorphism $s'$ from $\mathcal{H}=\H_{k+1}^N$ to $\H_{k+1}$. Consider the diagonal copy of $\H^d_0 \subset \cdots \subset \H^d_k \subset \H^d_{k+1} \subset G'^d$ in $\mathcal{G}'$. By abuse of notation we will also consider each of $s$ and $s'$ acting just on the diagonal.

If $|s'(\H^d_0)|=1$, by construction of $\mathcal{G}''$, we have that $s'(\H^d_1)$ is either an in-star or a loop, but the former is not possible as $\H^d_1$ is strongly connected. Iterating this argument we find that $|s'(\H^d_{k+1})|=1$, but this would mean $s'$ is uniformly mapping $\mathcal{H}_{k+1}$ to one vertex, which is impossible as $s'$ is non-trivial. As $1<|s'(\H^d_0)|<m$ is not possible either due to Lemma~\ref{lem:under-m}, we find that $|s'(\H^d_0)|=m$, and indeed $s'$ maps $\H^d_0$ to a copy of itself in $\H$ which we will call $\H'_0=i(\H^d_0)$ for some isomorphism~$i$. 

We now apply Lemma~\ref{lem:8} as well as our assumed endo-trivialities to derive that $s'$ in fact maps $\H^d_k$ by the isomorphism $i$ to a copy of itself in $\H_{k+1}$ which we will call $\H'_{k}$. 

We claim that $\mathrm{Spill}^+_{a_k}(\H_{k+1}[\H'_{k+1},i(\HC^d_{a_{k}})]) = V(\H_{k+1})$. Since $\lambda'$ is surjective on $\H_{k+1}^+$, this is enforced explicitly by the $\Cylplus_{a_{k}}$ of the third stage. Thus, there exists a retraction $r$ from $\H_{k+1}$ to $\H'_k$.  Now $i^{-1} \circ r \circ s'$ gives the desired retraction of $\G$ to $\H_k$.   
\end{proof}

\begin{proposition}[General Case A-II]
\label{prop:main-general-a2}
Let $\H_{k+1}^+$ be some reflexive tournament that has an initial strongly connected component $\H_{k+1}$ that is non-trivial.
Let  $\H_{0},\H_{1}, \ldots, \H_k, \H_{k+1}$ be reflexive tournaments, the first $k+1$ of which have Hamilton cycles $\HC_{0},\HC_{1}, \ldots, \HC_{k}$, respectively,
so that $\H_0 \subseteq H_1 \subseteq \cdots \subseteq \H_k \subseteq \H_{k+1}$.
Suppose that $\H_0$, $(\H_1,\H_0)$, \ldots, $(\H_{k},\H_{k-1}), (\H_{k+1},\H_{k})$ are endo-trivial and that
\[
\begin{array}{lcl}
\mathrm{Spill}^+_{a_0}(\H_1[\H_0,\HC_{0}]) &= &V(\H_1) \\
\mathrm{Spill}^+_{a_1}(\H_2[\H_1,\HC_{1}]) &= &V(\H_2) \\
\hspace*{3mm} \vdots &\vdots &\hspace*{3mm}\vdots \\
\mathrm{Spill}^+_{a_{k-1}}(\H_{k}[\H_{k-1},\HC_{k-1}]) &= &V(\H_k)\\
\mathrm{Spill}^+_{a_{k}}(\H_{k+1}[\H_{k},\HC_{k}]) &= &V(\H_{k+1})
\end{array}
\]
Then {\sc $\H_{k+1}$-Retraction} can be polynomially reduced to $\QCSP(\H^+_{k+1})$.
\end{proposition}
\begin{proof}
Let $n=a_{k+1}=|V(\H_{k+1})|$ and let $a_k, \ldots, a_0$ be the cardinalities of $|V(\H_k)|,\ldots,|V(\H_0|$, respectively.
Let $\G$ be an instance of {\sc $\H_{k+1}$-Retraction}.
We build an instance  $\phi$ of $\QCSP(\H^+_{k+1})$
in the following fashion. Consider all possible functions $\lambda:[n]\rightarrow V(\H_{k+1})$. For some such $\lambda$, let $\mathcal{G}(\lambda)$ be the graph enriched with constants $c_1,\ldots,c_n$ where these are interpreted over some subset of $V(\H_{k+1})$ according to $\lambda$ in the natural way (acting on the subscripts).

Let $\mathcal{G}=\bigotimes_{\lambda \in {V(\H_{k+1})}^{[n]}} \mathcal{G}(\lambda)$. Let $\G^d$, $\H^d_{k+1}$,  $\H^d_k$, \ldots, $\H^d_0$ be the diagonal copies of $\G$, $\H_{k+1},\H_k,\ldots,\H_0$ in $\mathcal{G}$. Let $\mathcal{H}_{k+1}$ be the subgraph of $\mathcal{G}$ induced by $V(\H_{k+1}) \times \cdots \times V(\H_{k+1})$. Note that the constants $c_1,\ldots,c_n$ live in $\mathcal{H}_{k+1}$. 


Now build $\mathcal{G}'$ from $\mathcal{G}$ by the following procedure.
For each $i \in [k+1]$, and $v \in V(\H^d_{i}) \setminus V(\H^d_{i-1})$, add a copy of $\Cylplus_{a_{i-1}}$, where $v$ is identified with the vertex $x$ that is at the end of the out-edge pendant on the top copy of $\DC^*_{a_{i-1}}$ in $\Cylplus_{a_{i-1}}$ and the bottom copy of $\DC^*_{i-1}$ is to be identified with $\H_{i-1}$ according to the identity map of $\DC^*_{a_{i-1}}$ to $\HC_{i-1}$.

Now build $\mathcal{G}''$ by adding an edge from each vertex $c_i$ to a new vertex $d_i$ (for each $i \in [n]$). 

Finally, build $\phi$ from the canonical query of $\mathcal{G}''$, where we additionally turn the vertices $d_1,\ldots,d_n$ to outermost universal variables $z_1,\ldots,z_n$. Then existentially quantify all remaining constants and vertices innermost. Finally, restrict all except the universal variables to be in $V(\H_{k+1})$.

First suppose that~$\G$ retracts to~$\H_{k+1}$ by $r$. Let $\lambda'$ be some assignment of the universal variables $z_1,\ldots,z_n$ of $\phi$ to $\H_{k+1}^+$ and choose $y_1,\ldots,y_n$ backwards-adjacent to these in $\H_{k+1}$, mapped by $\lambda$. To prove $\phi \in \QCSP(\H^+_{k+1})$ it suffices to prove that there is a homomorphism from $\mathcal{G}'$ to $\H^+_{k+1}$ that extends $\lambda$ and for this it suffices to prove that there is a homomorphism $h$ from $\mathcal{G}$ that extends $\lambda$. The extension of the latter to the former will always be possible due to the spill assumptions.

Henceforth let us consider the homomorphic image of $\mathcal{G}$ that is $\mathcal{G}(\lambda)$. To prove $\phi \in \QCSP(\H^+_{k+1})$ it suffices to prove that there is a homomorphism from $\mathcal{G}(\lambda)$ to $\H_{k+1}$ that extends $\lambda$. Note that it will be sufficient to prove that $\G$ retracts to $\H_{k+1}$. Well this was our original assumption so we are done.  

Suppose now $\phi \in \QCSP(\H^+_{k+1})$.  Choose some surjection for $\lambda'$ mapping $z_1,\ldots,z_n$ to $\H_{k+1}$. Choose some $y_1,\ldots,y_n$ backwards-adjacent to these (and therefore in $\H_{k+1}$) and let this be the map $\lambda$. Note that it is not possible for all $y_1,\ldots,y_n$ to be evaluated as a single vertex as $\H_{k+1}$ is strongly connected. Recall $N=|V(\H)^{[n]}|$.  The evaluation of the existential variables that witness $\phi \in \QCSP(\H^+_{k+1})$ induces a non-trivial homomorphism $s$ from $\mathcal{G}$ to $\H_{k+1}$ which contains within it a non-trivial homomorphism $s'$ from $\mathcal{H}_{k+1}=\H_{k+1}^N$ to $\H_{k+1}$. Consider the diagonal copy of $\H^d_0 \subset \H^d_1 \subset \cdots H^d_{k+1}$ in $\mathcal{G}$. By abuse of notation we will also consider each of $s$ and $s'$ acting just on the diagonal.

If $|s'(\H^d_0)|=1$ we deduce that $s'(\H^d_1)$ is either an in-star or a loop, but the former is not possible as $\H^d_1$ is strongly connected. Iterating this argument we find that $|s'(\H^d_{k+1})|=1$, but this would mean $s'$ is uniformly mapping to one vertex, which is impossible as $s'$ is non-trivial.
As $1<|s'(\H^d_0)|<m$ is not possible either due to Lemma~\ref{lem:under-m}, we find that $|s'(\H^d_0)|=m$, and indeed $s'$ maps $\H^d_0$ to a copy of itself in $\H$ which we will call $\H'_0=i(\H^d_0)$ for some isomorphism~$i$.

We now apply Lemma~\ref{lem:8} as well as our assumed endo-trivialities to derive that $s'$ in fact maps $\H^d_k$ by the isomorphism $i$ to a copy of itself in $\H_{k+1}$, which we will call $\H'_{k}$. Now we can deduce, via Lemma~\ref{lem:7}, that $h(\H^d_{k+1})$ is an automorphism of $\H_{k+1}$, which we call $\alpha$.
The required retraction from $\G$ to 
$\H_{k+1}$ 
is now given by $\alpha^{-1} \circ s'$.
\end{proof}

The proof of the following is exactly as that for Corollary~\ref{cor:strongly-connected1} modulo $\mathrm{Spill}$ becoming $\mathrm{Spill}^+$.
\begin{corollary}
Let $\H$ be a reflexive tournament with an initial strongly connected component that is non-trivial. Then $\QCSP(\H)$ is \NP-hard.
\label{cor:initial-full}
\end{corollary}

\section{The Proof of the NL Cases of the Dichotomy}\label{s-nl}
 
A particular role in the tractable part of our dichotomy will be played by $\TT_2$, the reflexive transitive $2$-tournament, which has vertex set $\{0,1\}$ and edge set $\{(0,0),(0,1),(1,1)\}$.

\begin{lemma}
\label{lem:sur-hom}
Let $\H=\H_1\Rightarrow \cdots \Rightarrow \H_n$ be a reflexive tournament on $m+2$ vertices with $V(\H_1)=\{s\}$ and $V(\H_n)=\{t\}$. Then there exists a surjective homomorphism from $(\TT_2)^m$ to $\H$.
\end{lemma}
\begin{proof}
Build a surjective homomorphism $f$ from $(\TT_2)^m$ to $\H$ in the following fashion. Let $\overline{x}_i$ be the $m$-tuple which has $1$ in the $i$th position and $0$ in all other positions. For $i\in [m]$, let $f$ map $\overline{x}_i$ to $i$. Let $f$ map $(0,\ldots,0)$ to $s$ and everything remaining to $t$. 

By construction, $f$ is surjective. To see that $f$ is a homomorphism, 
let $((y_1,\ldots,y_m),$ $(z_1,\ldots,z_m)) \in E((\TT_2)^m)$, which is the case exactly when
 $y_i\leq z_i$ for all $i \in [m]$. Let $f(y_1,\ldots,y_m)=u$ and $f(z_1,\ldots,z_m)=v$.
First suppose that $y_1,\ldots,y_m$ are all $0$. Then $u=s$. As $s$ has an out-edge to every vertex of $\H$, we find that $(u,v)\in E(\H)$.
Now suppose that $y_1,\ldots,y_m$ contains a single~$1$. If $(y_1,\ldots,y_m)=(z_1,\ldots,z_m)$, then $u=v$. As $\H$ is reflexive, we find that $(u,v)\in \H$. If $(y_1,\ldots,y_m)\neq (z_1,\ldots,z_m)$, then $v=t$. As $t$ has an in-edge from every vertex of $\H$, we find that $(u,v)\in E(\H)$.
Finally suppose that $y_1,\ldots,y_m$ contains more than one $1$. Then $u=v=t$. As $\H$ is reflexive, we find that $(u,v)\in E(\H)$.
\end{proof}

We also need the following lemma, which follows from combining some known results.

\begin{lemma}
\label{lem:NL-base}
If $\H$ is a transitive reflexive tournament then $\QCSP(\H)$ is in \NL.
\end{lemma}
\begin{proof}
It is noted in \cite{STACStoTOCT} that $\H$ has the ternary median operation as a polymorphism. 
It follows from well-known results (e.g. in \cite{hubie-sicomp,DK08}) that $\QCSP(\H)$ is in \NL. \textcolor{black}{Specifically, one can apply Theorem 5.16 from \cite{hubie-sicomp} to reduce  $\QCSP(\H)$ to an ensemble of instances of $\CSP(\H)$, which may also reference constants, each of which can be solved in $\NL$ by Corollary 4 from \cite{DK08}. Each of these instances may be solved independently and the ensemble is polynomial in number, hence the whole procedure can be accomplished in \NL.}
\end{proof}
The other tractable cases are more interesting.

We are now ready to prove the main result of this section.

\begin{theorem}\label{t-easy}
Let $\H=\H_1\Rightarrow \cdots \Rightarrow \H_n$ be a reflexive tournament. If $|V(H_1)|=|V(H_n)|=1$, then $\QCSP(\H)$ is in \NL.
\end{theorem}
\begin{proof}
Let $|V(\H)|=m+2$ for some $m\geq 0$.
By Lemma \ref{lem:sur-hom}, there exists a surjective homomorphism from $(\TT_2)^m$ to $\H$. 
 There exists also a surjective homomorphism from $\H$ to $\TT_2$; we map $s$ to $0$ and all other vertices of $\H$ to $1$.
 It follows from \textcolor{black}{Theorem 3.4 in} \cite{LMCS2015} that $\QCSP(\H)=\QCSP(\TT_2)$ meaning we may consider the latter problem.
 We note that $\TT_2$ is a transitive reflexive tournament. Hence, we may appply Lemma \ref{lem:NL-base}.
\end{proof}

\section{Final result and remarks}\label{s-con}

We are now in a position to prove our main dichotomy theorem.
\begin{theorem}\label{t-dichotomy}
Let $\H=\H_1\Rightarrow \cdots \Rightarrow \H_n$ be a reflexive tournament. If $|V(H_1)|=|V(H_n)|=1$, then $\QCSP(\H)$ is in \NL; otherwise it is \NP-hard.
\end{theorem}
\begin{proof}
The \NL\ case
follow from Theorem~\ref{t-easy}. The \NP-hard cases follow from Corollary~\ref{cor:strongly-connected1} and Corollary~\ref{cor:initial-full}, bearing in mind the case with a non-trivial final strongly connected component is dual to the case   with a non-trivial initial strongly connected component  (map edges $(x,y)$ to $(y,x)$).
\end{proof}
Theorem~\ref{t-dichotomy} resolved the open case in Table~\ref{fig:context}. \textcolor{black}{It is difficult to position this result in the overall classification program for finite-domain QCSPs save to say that our methods are tailored, indeed specialised, to reflexive tournaments. It is not clear that they can be applied easily to different or wider classes (in this vein we return to mixed-type tournaments below). Since complexities outside of P, NP-complete and Pspace-complete were discovered for QCSPs in \cite{ZhukM20}, for example co-NP-complete, DP-complete and $\Theta_2^{\mathrm{P}}$, the whole classification task has been thrown wide open. Classes such as that of reflexive tournaments might provide comfort, as it is doubtful such monstrous complexities could be found here. Though, we cannot be sure, with our lacuna between NP-hard and Pspace-complete.} 

Recall that the results for the irreflexive tournaments in this table were all proven in a more general setting, namely for irreflexive semicomplete graphs. \textcolor{black}{One} natural direction for future research is to determine a complexity dichotomy for \QCSP\ and $\mathrm{SCSP}$
for reflexive semicomplete graphs. We leave this as an interesting open direction.

{\color{black} The task of promoting our \NP-hardness results to \Pspace-complete, while using the same method, seems to require corresponding \Pspace-hardness results for reflexive tournaments with constants. If $\QCSP^c(\H)$ were \Pspace-complete, for $\H$ a non-trivial reflexive strongly connected tournament, then likely our \NP-hardness results, for the similar class of graphs, would easily rise to \Pspace-complete. The cases that are not strongly connected require additional arguments, and perhaps even a different method.}

{\color{black} Mixed-type tournaments, where some vertices are reflexive and others irreflexive, are well-understood algebraically \cite{Wires15}. Indeed, from this paper there follows a complexity dichotomy for $\CSP^c(\H)$ where $\H$ is a mixed-type tournament. Furthermore, $\CSP(\H)$ is either trivial or $\H$ is an irreflexive tournament, so the complexity dichotomy for $\CSP(\H)$ is also known. Though many of our supporting lemmas hold for mixed-type tournaments, some do not. For example, Lemma~\ref{lem:uniformly-z} fails for the transitive $2$-tournament $\mathrm{TT}_2$ in which one vertex is a self-loop and the other is not. To extend our classification to mixed-type tournaments thus requires still some work.}


\begin{thebibliography}{10}

\bibitem{BHM88}
J{\o}rgen Bang{-}Jensen, Pavol Hell, and Gary MacGillivray.
\newblock The complexity of colouring by semicomplete digraphs.
\newblock {\em {SIAM} Journal on Discrete Mathematics}, 1(3):281--298, 1988.

\bibitem{BKM12}
Manuel Bodirsky, Jan K{\'{a}}ra, and Barnaby Martin.
\newblock The complexity of surjective homomorphism problems - a survey.
\newblock {\em Discrete Applied Mathematics}, 160(12):1680--1690, 2012.

\bibitem{BBCJK}
Ferdinand B{\"o}rner, Andrei~A. Bulatov, Hubie Chen, Peter Jeavons, and
  Andrei~A. Krokhin.
\newblock The complexity of constraint satisfaction games and {QCSP}.
\newblock {\em Inf. Comput.}, 207(9):923--944, 2009.

\bibitem{Bu17}
Andrei~A. Bulatov.
\newblock A dichotomy theorem for nonuniform {CSP}s.
\newblock In {\em 58th {IEEE} Annual Symposium on Foundations of Computer
  Science, {FOCS} 2017, Berkeley, CA, USA, October 15-17, 2017}, pages
  319--330, 2017.

\bibitem{JBK05}
Andrei~A. Bulatov, Peter Jeavons, and Andrei~A. Krokhin.
\newblock Classifying the complexity of constraints using finite algebras.
\newblock {\em {SIAM} Journal on Computing}, 34(3):720--742, 2005.

\bibitem{Ca59}
Paul Camion.
\newblock Chemins et circuits hamiltoniens de graphes complets.
\newblock {\em Comptes Rendus de l'Acad\'emie des Sciences Paris},
  249:2151--2152, 1959.

\bibitem{hubie-sicomp}
Hubie Chen.
\newblock The complexity of quantified constraint satisfaction: Collapsibility,
  sink algebras, and the three-element case.
\newblock {\em SIAM J. Comput.}, 37(5):1674--1701, 2008.
\newblock \href {https://doi.org/http://dx.doi.org/10.1137/060668572}
  {\path{doi:http://dx.doi.org/10.1137/060668572}}.

\bibitem{LMCS2015}
Hubie Chen, Florent~R. Madelaine, and Barnaby Martin.
\newblock Quantified constraints and containment problems.
\newblock {\em Logical Methods in Computer Science}, 11(3), 2015.
\newblock Extended abstract appeared at LICS 2008. This journal version
  incorporates principal part of CP 2012 {\it Containment, Equivalence and
  Coreness from CSP to QCSP and Beyond}.
\newblock URL: \url{http://dx.doi.org/10.2168/LMCS-11(3:9)2015}, \href
  {https://doi.org/10.2168/LMCS-11(3:9)2015}
  {\path{doi:10.2168/LMCS-11(3:9)2015}}.

\bibitem{DK08}
V{\'{\i}}ctor Dalmau and Andrei~A. Krokhin.
\newblock Majority constraints have bounded pathwidth duality.
\newblock {\em European Journal of Combinatorics}, 29(4):821--837, 2008.

\bibitem{TOCL17}
Petar Dapic, Petar Markovic, and Barnaby Martin.
\newblock Quantified constraint satisfaction problem on semicomplete digraphs.
\newblock {\em {ACM} Trans. Comput. Log.}, 18(1):2:1--2:47, 2017.
\newblock \href {https://doi.org/10.1145/3007899} {\path{doi:10.1145/3007899}}.

\bibitem{FV98}
Tom{\'{a}}s Feder and Moshe~Y. Vardi.
\newblock The computational structure of monotone monadic {SNP} and constraint
  satisfaction: {A} study through datalog and group theory.
\newblock {\em {SIAM} Journal on Computing}, 28(1):57--104, 1998.

\bibitem{GPS12}
Petr~A. Golovach, Dani{\"{e}}l Paulusma, and Jian Song.
\newblock Computing vertex-surjective homomorphisms to partially reflexive
  trees.
\newblock {\em Theoretical Computer Science}, 457:86--100, 2012.

\bibitem{KolaitisVardiBook05}
P.~G. Kolaitis and M.~Y. Vardi.
\newblock {\em Finite Model Theory and Its Applications (Texts in Theoretical
  Computer Science. An EATCS Series)}, A logical Approach to Constraint
  Satisfaction.
\newblock Springer-Verlag New York, Inc., 2005.

\bibitem{La06}
Benoit Larose.
\newblock Taylor operations on finite reflexive structures.
\newblock {\em International Journal of Mathematics and Computer Science},
  1(1):1--26, 2006.

\bibitem{STACStoTOCT}
Benoit Larose, Barnaby Martin, and Dani{\"{e}}l Paulusma.
\newblock Surjective {H}-colouring over reflexive digraphs.
\newblock {\em {TOCT}}, 11(1):3:1--3:21, 2019.
\newblock \href {https://doi.org/10.1145/3282431} {\path{doi:10.1145/3282431}}.

\bibitem{LZ05}
Benoit Larose and L\'aszl\'o Z\'adori.
\newblock Finite posets and topological spaces in locally finite varieties.
\newblock {\em Algebra Universalis}, 52(2):119--136, 2005.

\bibitem{Martin11}
Barnaby Martin.
\newblock {QCSP} on partially reflexive forests.
\newblock In {\em Principles and Practice of Constraint Programming - {CP} 2011
  - 17th International Conference, {CP} 2011, Perugia, Italy, September 12-16,
  2011. Proceedings}, pages 546--560, 2011.
\newblock \href {https://doi.org/10.1007/978-3-642-23786-7\_42}
  {\path{doi:10.1007/978-3-642-23786-7\_42}}.

\bibitem{CiE2006}
Barnaby Martin and Florent~R. Madelaine.
\newblock Towards a trichotomy for quantified \emph{H}-coloring.
\newblock In {\em Logical Approaches to Computational Barriers, Second
  Conference on Computability in Europe, CiE 2006, Swansea, UK, June 30-July 5,
  2006, Proceedings}, pages 342--352, 2006.
\newblock \href {https://doi.org/10.1007/11780342\_36}
  {\path{doi:10.1007/11780342\_36}}.

\bibitem{Vi13}
Narayan Vikas.
\newblock Algorithms for partition of some class of graphs under compaction and
  vertex-compaction.
\newblock {\em Algorithmica}, 67(2):180--206, 2013.

\bibitem{Vi17}
Narayan Vikas.
\newblock Computational complexity of graph partition under vertex-compaction
  to an irreflexive hexagon.
\newblock In {\em 42nd International Symposium on Mathematical Foundations of
  Computer Science, {MFCS} 2017, August 21-25, 2017 - Aalborg, Denmark}, pages
  69:1--69:14, 2017.

\bibitem{Wires15}
Alexander Wires.
\newblock Dichotomy for finite tournaments of mixed-type.
\newblock {\em Discrete Mathematics}, 338(12):2523--2538, 2015.
\newblock \href {https://doi.org/10.1016/j.disc.2015.06.024}
  {\path{doi:10.1016/j.disc.2015.06.024}}.

\bibitem{Zh17}
Dmitriy Zhuk.
\newblock A proof of {CSP} dichotomy conjecture.
\newblock In {\em 58th {IEEE} Annual Symposium on Foundations of Computer
  Science, {FOCS} 2017, Berkeley, CA, USA, October 15-17, 2017}, pages
  331--342, 2017.

\bibitem{zhuk2020norainbow}
Dmitriy Zhuk.
\newblock No-rainbow problem and the surjective constraint satisfaction
  problem, 2020.
\newblock To appear LICS 2021.
\newblock \href {http://arxiv.org/abs/2003.11764} {\path{arXiv:2003.11764}}.

\bibitem{Zh20}
Dmitriy Zhuk.
\newblock A proof of the {CSP} dichotomy conjecture.
\newblock {\em J. {ACM}}, 67(5):30:1--30:78, 2020.
\newblock \href {https://doi.org/10.1145/3402029} {\path{doi:10.1145/3402029}}.

\bibitem{ZhukM20}
Dmitriy Zhuk and Barnaby Martin.
\newblock {QCSP} monsters and the demise of the {Chen Conjecture}.
\newblock In Konstantin Makarychev, Yury Makarychev, Madhur Tulsiani, Gautam
  Kamath, and Julia Chuzhoy, editors, {\em Proccedings of the 52nd Annual {ACM}
  {SIGACT} Symposium on Theory of Computing, {STOC} 2020, Chicago, IL, USA,
  June 22-26, 2020}, pages 91--104. {ACM}, 2020.
\newblock \href {https://doi.org/10.1145/3357713.3384232}
  {\path{doi:10.1145/3357713.3384232}}.

\end{thebibliography}

\section*{Acknowledgements}

We are grateful to several referees for careful reading of the paper and good advice.

\end{document}